\documentclass[a4paper,USenglish]{lipics-v2018}
\usepackage{microtype}%if unwanted, comment out or use option "draft"
\usepackage[ruled,vlined,linesnumbered]{algorithm2e}
\usepackage[inline]{enumitem}
\usepackage{subcaption}
\usepackage{graphicx}
\graphicspath{{./Graphics/}}%helpful if your graphic files are in another directory

\nolinenumbers

\title{On Restricted Disjunctive Temporal Problems: Faster Algorithms and Tractability Frontier}
\titlerunning{On RDTPs: Faster Algorithms and Tractability Frontier}

\author{Carlo Comin}{Department of Computer Science, University of Verona, Italy}{carlo.comin.86@gmail.com}{https://orcid.org/0000-0001-5748-2029}{}
\author{Romeo Rizzi}{Department of Computer Science, University of Verona, Italy}{romeo.rizzi@univr.it}{https://orcid.org/0000-0002-2387-0952}{}
\authorrunning{C. Comin, R. Rizzi}

\Copyright{Carlo Comin, Romeo Rizzi}

%%please refer to \url{http://www.acm.org/about/class/ccs98-html}}%
% mandatory: Please provide 1-5 keywords

\EventEditors{}
\EventNoEds{3}
\EventLongTitle{Preprint}
\EventShortTitle{Preprint}
\EventAcronym{}
\EventYear{2018}
\EventDate{}
\EventLocation{}
\EventLogo{}
\SeriesVolume{}
\EventYear{2018}
\ArticleNo{} % “LIPIcs article number” (=<article-no>) goes here!

\theoremstyle{plain}
\newtheorem{mytheorem}{\textbf{Theorem}}
\newtheorem{proposition}{\textbf{Proposition}}
\theoremstyle{definition}
\newtheorem{mydefinition}{Definition}
\newtheorem{mylemma}{Lemma}

\definecolor{colorbox}{RGB}{253,199,17}

\AtBeginDocument{%

  \renewcommand{\algorithmautorefname}{Algo.\negthinspace}
}

\newcommand{\GHTP}[1][]{\textsc{General-HyTP#1}\xspace}
\newcommand{\GHTN}[1][]{\textsc{General-HyTN#1}\xspace}

\newcommand{\THTP}[1][]{\textsc{Tail-HyTP#1}\xspace}

\newcommand{\TONE}{\texttt{t}_1\xspace}
\newcommand{\TTWO}{\texttt{t}_2\xspace}
\newcommand{\TTHREE}{\texttt{t}_3\xspace}

\let\oldnl\nl
\newcommand{\nonl}{\renewcommand{\nl}{\let\nl\oldnl}}

\usepackage{tikz}
\usepackage{pgfplots}
\usetikzlibrary{positioning,backgrounds,decorations,arrows,shapes,fit,calc, patterns}
\pgfplotsset{compat=newest}%
\tikzset{>=latex}
\tikzstyle{node}=[circle,draw,inner sep=2pt,transform shape,minimum size=1.75em]
\tikzstyle{smallLabel}=[font=\sffamily\footnotesize,inner sep=1pt,outer sep=1pt,transform shape]
\tikzstyle{timeLabel}=[smallLabel,midway,transform shape]
\tikzstyle{multiHead}=[draw,->,dashed,rounded corners]
\tikzstyle{RLink} = [draw,->,rounded corners]
\tikzstyle{CLink} = [draw,->,double,rounded corners]
\tikzstyle{Connector} = [diamond, draw, minimum height=20pt,minimum width=20pt, fill=gray!20,align=center,thick]
\tikzstyle{Task} = [draw, minimum height=20pt,minimum width=40pt, fill=gray!20,align=center,thick]
\tikzstyle{Arc} = [draw,thick,->,rounded corners]

\tikzstyle{multiHead}=[dashed,transform shape]
\tikzstyle{multiTail}=[dotted,thick,transform shape]

\usepackage{graphicx}
\usepackage{cancel}
\usepackage{newfloat}

\newcommand{\makecal}[2]{\newcommand{#1}{\mathcal{#2}}}
\makecal{\calc}{C}
\makecal{\cale}{E}
\makecal{\calh}{H}
\def\H{{\cal H}}
\makecal{\cali}{I}
\makecal{\cala}{A}
\makecal{\lvs}{L}
\makecal{\caln}{N}
\makecal{\calo}{O}
\makecal{\calot}{OT}
\makecal{\caldt}{DT}

\makecal{\calp}{P}
\makecal{\calq}{Q}
\makecal{\cals}{S}
\makecal{\calt}{T}
\makecal{\calv}{V}
\makecal{\calcl}{L}

\newcommand{\ie}{i.e.,\xspace}
\newcommand{\eg}{e.g.,\xspace}

\newcommand{\wrt}{w.r.t.\xspace}
\renewcommand{\iff}{\textit{iff}\xspace}
\newcommand{\TTP}{\texttt{$\text{t}_2\text{DTP()}$}\xspace}

\newcommand{\figref}[1]{Fig.~\ref{#1}}

\newcommand{\reali}{\ensuremath{\mathbb{R}}}
\newcommand{\RR}{\ensuremath{\mathbb{R}}}
\newcommand{\interi}{\ensuremath{\mathbb{Z}}}

\newcommand{\T}{\mathcal{T}}
\newcommand{\F}{\mathcal{F}}
\newcommand{\C}{\mathcal{C}}

\newcommand{\A}{\mathcal{A}}

\newcommand{\N}{\mathcal{N}}

\newcommand{\NP}{\text{NP}}
\newcommand{\coNP}{\text{co-NP}}

\newcommand{\STN}{STN\xspace}
\newcommand{\STP}{STP\xspace}
\newcommand{\RDTN}{RDTN\xspace}
\newcommand{\DTP}{DTP\xspace}
\newcommand{\RDTP}{RDTP\xspace}
\newcommand{\RDTPc}{\texttt{RDTP()}\xspace}
\newcommand{\TTHTP}{$\texttt{t}_2\texttt{HyTP()}$\xspace}
\newcommand{\HTN}{HyTN\xspace}

\newcommand{\HHTP}{\textsc{head-HyTP}\xspace}
\newcommand{\TTTHTNC}{\textsc{head-$\TTHREE$HyTP}\xspace}
\newcommand{\TTTTTNC}{\textsc{tail-$\TTHREE$HyTP}\xspace}
\newcommand{\TTTtwoHTNC}{\textsc{head-$\TTWO$HyTP}\xspace}
\newcommand{\TTTtwoTTNC}{\textsc{tail-$\TTWO$HyTP}\xspace}

\tikzstyle{guardRange}=[above,scale=.8,outer sep=0pt,inner sep=1.5pt,xshift=-2pt,transform shape,minimum width=\mingnat]
\newlength{\gnat}
\newlength{\mingnat}
\newcommand{\setgnat}[1]{%
	\settowidth{\gnat}{\scalebox{1}{\normalsize\ensuremath{#1}}}%
	\ifthenelse{\lengthtest{\gnat<\mingnat}}{\setlength{\gnat}{\mingnat}}{}%
}

\begin{document}

\maketitle

\begin{abstract}
In 2005 T.K.S. Kumar studied the Restricted Disjunctive Temporal Problem (RDTP),
	a restricted but very expressive class of Disjunctive Temporal Problems (DTPs).
An RDTP comes with a finite set of temporal variables, and a finite set of temporal constraints
 each of which can be either one of the following three types:
($\TONE$) two-variable linear-difference simple constraint;
($\TTWO$) single-variable disjunction of many interval constraints;
($\TTHREE$) two-variable disjunction of two interval constraints only.
Kumar showed that RDTPs are solvable in deterministic strongly polynomial time by reducing them to the Connected Row-Convex (CRC) constraints satisfaction problem,
also devising a faster randomized algorithm. Instead,
the most general form of DTPs allows for multi-variable disjunctions of many interval constraints and it is \textsc{NP}-complete.

This work offers a deeper comprehension on the tractability of RDTPs, leading to an elementary deterministic strongly polynomial time algorithm for them,
significantly improving the asymptotic running times of all the previous deterministic and randomized solutions.
The result is obtained by reducing RDTPs to the Single-Source Shortest Paths (SSSP) and the 2-SAT problem (jointly), instead of reducing to CRCs.
In passing, we obtain a faster (quadratic time) algorithm for RDTPs having only $\{\TONE, \TTWO\}$-constraints and no $\TTHREE$-constraint.
As a second main contribution, we study the tractability frontier of solving RDTPs blended with Hyper Temporal Networks (\HTN{s}),
a disjunctive strict generalization of Simple Temporal Networks (\STN{s}) based on hypergraphs:
	we prove that solving temporal problems having only $\TTWO$-constraints and either only multi-tail or only multi-head hyperarc-constraints
lies in $\textsc{NP}\cap\textsc{\text{co-}NP}$ and admits deterministic pseudo-polynomial time algorithms;
on the other hand, problems having only $\TTHREE$-constraints and either only multi-tail or only multi-head hyperarc-constraints~turns~out~strongly~\textsc{NP}-complete.
\end{abstract}

\section{Introduction}\label{sect.introduction}
Expressive and efficient temporal reasoning is essential to a number of areas in Artificial Intelligence~(AI)~\cite{KOUBARAKIS2006,PANI200155,Schwalb1998}.
Over the past few years, many constraint-based formalisms have been developed to represent
and reason about time in automated planning and temporal scheduling~\cite{Dechter2003,Nau2004}.
We begin by recalling the Disjunctive Temporal Problem (DTP)~\cite{Oddi2000,STERGIOU2008,TSAMARDINOS2003}.
The general form of a DTP being, given a finite set $\T=\{X_0,X_1,\ldots, X_N\}$ of temporal variables (\ie time-points),
to schedule them on the real line in such a way as to satisfy a prescribed finite set $\C$ of temporal constraints over $\T$.
Every constraint $c_i\in\C$ is a disjunction of the form $s_{(i,1)}\vee s_{(i,2)}\vee \cdots \vee s_{(i,T_i)}$,
where every $s_{i,j}$ is a simple temporal constraint of the form
$(l_{i,j}\leq X_{\beta_{i,j}} - X_{\alpha_{i,j}} \leq u_{i,j})$ for some integers $0\leq \alpha_{i,j}, \beta_{i,j} \leq N$ and reals $l_{i,j}, u_{i,j}$.

Although DTPs are expressive enough to capture many tasks in automated planning and temporal scheduling, they are \textsc{NP}-complete~\cite{STERGIOU2008}.
The principal direct approach taken to solve DTPs has been to convert the original
problem into one of selecting a disjunct from each constraint~\cite{STERGIOU2008,TSAMARDINOS2003},
then to check whether the set of selected disjuncts forms a consistent Simple Temporal Problem (STP)~\cite{DechterMP91}.
This can be done in strongly polynomial time by computing single-source shortest paths (\eg with the Bellman-Ford's algorithm~\cite{Bellman58}).
Under this prospect, of course the prohibitive complexity of solving DTPs comes from the fact that there are exponentially many disjunct combinations~possible.

In~\cite{Kumar2006,Kumar05}, T.K.S. Kumar studied the Restricted Disjunctive Temporal Problem (RDTP),
  a tractable subclass of DTPs strictly including the classical and well established STPs~\cite{DechterMP91}.
In RDTPs, each constraint can be either one of the following three types:
($\TONE$) $(Y-X\leq w)$, for $w$ real (a simple temporal difference-constraint);
($\TTWO$) $(l_1\leq X\leq u_1)\vee \cdots \vee (l_k\leq X\leq u_k)$, for $l_i,u_i$ reals (a single-variable disjunction of many interval-constraints);
($\TTHREE$) $(l_1\leq X\leq u_1) \vee (l_2\leq Y\leq u_2)$, for $l_i,u_i$ reals (a two-variable disjunction of two interval-constraints).

It was shown in~\cite{Kumar05} that RDTPs are solvable in deterministic strongly polynomial
time by reducing them to the Connected Row-Convex (CRC)~\cite{DEVILLE1999} constraint satisfaction problem,
  faster randomized algorithms were also proposed.
CRC constraints generalize many other known tractable classes of constraints like 2-SAT,
  implicational, and binary integer-weighted linear constraints~\cite{DEVILLE1999}.
Particularly, Kumar's deterministic algorithm for solving RDTPs works
by reducing them into binary Constraint Satisfiability Problems (CSPs) over meta-variables representing
$\TTWO$ or $\TTHREE$ constraints, meanwhile showing that such binary constraints are indeed CRC constraints,
finally exploiting the algorithmic tractability of CRC constraints.

An instantiation of a consistency checking algorithm (\eg~\cite{DEVILLE1999}) that further exploits the structure of CRC constraints
leads to a time complexity of $O\big((|\C_{\TTWO}| + |\C_{\TTHREE}|)^3\cdot d^2_{\max} + |\T|\cdot |\C_{\TONE}|\cdot (|\C_{\TTWO}| + |\C_{\TTHREE}|)^2\big)$,
where $\C_{\TONE, \TTWO, \TTHREE}$ is the set of ${\TONE},{\TTWO},{\TTHREE}$ constraints (respectively),
and $d_{\max}$ is the maximum number of disjuncts possible per single constraint~\cite{Kumar05}.
Randomization reduces the running time to
$O\big((|\C_{\TTWO}|+|\C_{\TTHREE}|)^2 \cdot d_{\max}^2 \cdot \delta + |\T|\cdot |\C_{\TONE}|\cdot (|\C_{\TTWO}| + |\C_{\TTHREE}|)^2\big)$,
where $\delta$ is the degree of the CRC network (\ie the maximum number of constraints any variable participates into)~\cite{Kumar05}.

Notable applications of RDTPs include solving STP{s} with Taboo Regions, cfr.~\cite{Kumar2013}.

\textit{Contributions.} This work offers a deeper comprehension on the tractability of RDTPs,
leading to elementary deterministic strongly polynomial time algorithms,
significantly improving the asymptotic running times of both the Kumar's deterministic and randomized solutions.
Our time complexity is $O\big(|\T|\cdot |\C_{\TONE}| + |\C_{\TTWO}|\cdot (|\C_{\TONE}| + |\T|\cdot \log |\T|) + |\T|\cdot d_{\C_{\TTWO}}\cdot|\C_{\TTHREE}| + |\C_{\TTHREE}|^2\big)$,
where $d_{\C_{\TTWO}}$ is the total number of disjuncts counting over all $\TTWO$-constraints.
Since $d_{\C_{\TTWO}}\leq d_{\max}\cdot |\C_{\TTWO}|$, this improves over all of the previous solutions.
The result is obtained by reducing RDTPs to the Single-Source Shortest Paths (SSSP) and the 2-SAT problem (jointly), instead of reducing to CRCs.
So the full expressive power of CRCs is not needed, binary linear and 2-SAT constraints are enough.
In passing, we obtain a faster (quadratic time) deterministic algorithm for solving temporal problems
having only $\{\TONE, \TTWO\}$-constraints and no ${\TTHREE}$-constraint.

As a second main contribution, we study the tractability frontier of RDTPs widened
with another kind of restricted disjunctive constraints, \ie Hyper Temporal Networks (\HTN{s})~\cite{CominPR17},
 a strict generalization of \STN{s} grounded on directed hypergraphs and introduced to overcome the limitation of considering only conjunctions of
constraints but maintaining a practical efficiency in the consistency check of the instances.
In a \HTN a single temporal multi-tail (or multi-head) hyperarc-constraint is defined as a set of two or more maximum delay (minimum anticipation, respectively)
constraints which is satisfied when at least one of these delay constraints is so.
We prove that solving temporal problems having only ${\TTWO}$-constraints and either only multi-tail or only multi-head hyperarc-constraints
 lies in $\textsc{NP}\cap\textsc{\text{co-}NP}$ and admits deterministic pseudo-polynomial time algorithms;
on the other hand, solving temporal problems having only ${\TTHREE}$-constraints and either only multi-tail or
only multi-head hyperarc-constraints turns out strongly \textsc{NP}-complete. See Table~1 below for a summary.
\begin{center}
  \begin{tabular}{ | c | c | c | c |}
    \textit{Problem} & \textit{Complexity} & \textit{Improved Time Bound} & \textit{Cfr.} \\ \hline
    $\TTWO$DTPs & \textsc{P} & $O\big(|\T|\cdot |\C_{\TONE}| + |\C_{\TTWO}|\cdot (|\C_{\TONE}|
      + |\T|\cdot \log |\T|) + |\T|\cdot d_{\C_{\TTWO}}\big)$ & Sect.~3 \\ \hline
    RDTPs & \textsc{P} & \begin{tabular}{@{}c@{}}$O\big(|\T|\cdot |\C_{\TONE}| + |\C_{\TTWO}|\cdot (|\C_{\TONE}| + |\T|\cdot \log |\T|)\, +$  \\
      $+\,|\T|\cdot d_{\C_{\TTWO}}\cdot|\C_{\TTHREE}| + |\C_{\TTHREE}|^2\big)$ \end{tabular} & Sect.~4 \\ \hline
    $\TTWO$HyTPs & $\textsc{NP}\cap \text{co\,-\textsc{NP}}$ & $O\big((|\T|+|\A|)\cdot m_{\A}\cdot W_{\A,\C_{\TTWO}}\big)$ & Sect.~6 \\ \hline
    $\TTHREE$HyTPs  &  $\textsc{NP}\text{-complete}$ & \emph{n.a. (exponential time)} & Sect.~5 \\ \hline
  \end{tabular}\label{tab:results}
\end{center}
\begin{center}
\colorbox{colorbox}{\rule{0pt}{3pt}\rule{3pt}{0pt}}\; \footnotesize\textbf{Table~\ref{tab:results}} Summary of main results. \normalsize
\end{center}

\section{Background}\label{sect.Background}
This section offers the basic background notions that are assumed in the rest of the paper,
let's start with Simple Temporal Networks (\STN{s}) and related problems (STP{s}), cfr.~\cite{Dechter2003,DechterMP91}.
\begin{mydefinition}[\STN{s}, \STP{s}~\cite{Dechter2003,DechterMP91}]
A {\em Simple Temporal Network} (\STN) is a pair $(\T,\C)$, where $\T$ is a set of
real-valued variables called {\em time-points,} and $\C$ is a set of linear real-weighted binary constraints over $\T$
  called {\em simple (or $\TONE$)} temporal constraints, each having the form:
\[(Y-X\leq w_{X,Y}),\text{ where } X,Y\in \T \text{ and } w_{X,Y}\in\reali.\]

An \STN is \textit{consistent} if it admits a \emph{feasible schedule},
\ie some $s: \T\mapsto \reali$ such that $s(Y) \leq s(X) + w_{X,Y}$ for all $(Y-X\leq w_{X,Y})\in \C$.
So the {\em Simple Temporal Problem} (\STP) is that of determining whether a given \STN is consistent or not.
\end{mydefinition}

Any \STN $\N=(\T,\C)$ can be seen as a directed weighted graph with vertices $\T$ and arc set $A_\C\triangleq \{(X,Y, w_{X,Y})\mid (Y-X\leq w_{X,Y})\in \C\}$.
So, a \emph{path} $p$ in $\N$ is any finite sequence of vertices $p=(v_0,v_1,\ldots, v_k)$ (some $k\geq 1$)
such that $(v_i, v_{i+1}) \in A_\C$ for every $i\in [0,k)\cap\interi$;
the total weight of $p$ is then $w_p\triangleq \sum_{i=0}^{k-1} w_{v_i, v_{i+1}}$.
A \emph{cycle} $C$ in $\N$ is any set of arcs $C\subseteq A_\C$ cyclically sequenced as $a_0, a_1, \ldots a_{\ell-1}$
where an head equals a tail, \ie $h(a_i) = t(a_j)$, iff $j=i+1 \mod \ell$; it is called a \emph{negative cycle} if $w(C) \leq 0$,
where $w(C)$ stands for $\sum_{a\in C} w_a$. A graph is called \textit{conservative} when it contains no negative cycle.
A \textit{schedule} is any function $f: \T \mapsto \reali$.
So the \textit{reduced weight} of an arc $a = (t,h,w_a)$ with respect to a schedule $f$ is defined as $w^{f}_a \triangleq w_a - f(h) + f(t)$.
A schedule $f$ is \textit{feasible} iff $w^{f}_a\geq 0$ for every $a\in A_\C$.

It is also worth noticing that, given two feasible schedules $s_1,s_2$ of any \STN,
the pointwise-minimum schedule, $s(u)\triangleq \min(s_1(u), s_2(u))$ $\forall u\in\T,$ is also feasible.
Indeed, among all of the possible feasible schedules of a given consistent \STN, it is natural to consider the \emph{least} feasible one;
  \ie $\hat{s}:\T\rightarrow\reali_{\geq 0}$ is the \emph{least feasible schedule} of \STN $\N$ if $\hat{s}$ is feasible for $\N$ and,
for any other non-negative feasible schedule $s'\geq 0$ of $\N$, it holds $\hat{s}(u)\leq s'(u)$ $\forall u\in \T$.

Remarkably, finding the least feasible schedule of an \STN takes polynomial time~\cite{Dechter2003}.
\begin{mytheorem}[\cite{Dechter2003}]\label{thm:main_stn}
Let $\N=(\T,\C)$ be an \STN. The \emph{Bellman-Ford} (BF) algorithm (cfr.~\cite{Bellman58}) produces in $O(|\T|\cdot |\C|)$ time:
either the least feasible schedule $\hat{s} :\T\rightarrow\reali_{\geq 0}$, in case $\N$ is consistent;
or a certificate that $\N$ is inconsistent in the form of a negative cycle.
Moreover, if the weights of the arcs are all integers, then the scheduling values of $\hat{s} $ are all integers too.
\end{mytheorem}

Concerning the BF algorithm itself, it's worth considering an improved variant of it that we call the \emph{Bellman-Ford Value-Iteration}~(BF-VI).
The basic idea of BF-VI is the same as the original BF algorithm in that each vertex is used as a candidate to relax its adjacent vertices.
The improvement is that instead of trying all vertices blindly, BF-VI maintains a queue $Q$ of candidate vertices
 and adds a vertex to $Q$ only if that vertex is relaxed.
A candidate vertex $v$ is extracted from $Q$ according to a fixed policy (\eg LIFO), and then the
  adjacent vertices of $v$ are possibly relaxed as usual and added to $Q$ (if they are not already in there, no repetitions are allowed).
This process repeats until no more vertex can be relaxed.

BF-VI serves us as a basic model, to be leveraged to design faster algorithms for RDTPs.

\subsection{Restricted Disjunctive Temporal Problems}
\begin{figure}
\begin{minipage}[t]{0.45\textwidth}
\begin{tikzpicture}[scale=0.75]
\begin{axis}[axis lines=middle, xtick={0,...,9}, y=16cm, ymajorticks=false, axis equal,grid=both, xlabel=$X_i$, every axis x label/.style={
    at={(ticklabel* cs:1)},
    anchor=west},
	every tick/.style={
        black,
        thick
      }]
\addplot[color=black, mark=] coordinates{(0,1) (0,-1)};
\addplot[color=black, mark=] coordinates{(1,1) (1,-1)};
\addplot [color=black,mark=,fill=gray,
                    fill opacity=0.5] coordinates {
            (0, 1)
            (1, 1)
            (1, -1)
	    (0,-1)
};
\addplot[color=black, mark=] coordinates{(2,1) (2,-1)};
\addplot[color=black, mark=] coordinates{(3,1) (3,-1)};
\addplot [color=black,mark=,fill=gray,
                    fill opacity=0.5] coordinates {
            (2, 1)
            (3, 1)
            (3, -1)
	    (2,-1)
};
\addplot[color=black, mark=] coordinates{(5,1) (5,-1)};
\addplot[color=black, mark=] coordinates{(7,1) (7,-1)};
\addplot [color=black,mark=,fill=gray,
                    fill opacity=0.5] coordinates {
            (5, 1)
            (7, 1)
            (7, -1)
	    (5,-1)
};
\addplot[color=black, mark=] coordinates{(8,1) (8,-1)};
\addplot[color=black, mark=] coordinates{(9,1) (9,-1)};
\addplot [color=black,mark=,fill=gray,
                    fill opacity=0.5] coordinates {
            (8, 1)
            (9, 1)
            (9, -1)
	    (8,-1)
};
\end{axis}
\end{tikzpicture}\caption{An example of a $\TTWO$-constraint:
	$(0\leq X_i\leq 1)\vee (2\leq X_i\leq 3)\vee (5\leq X_i\leq 7)\vee (8\leq X_i\leq 9)$.}\label{fig:type2}
\end{minipage}
\qquad
\begin{minipage}[t]{0.45\textwidth}
\begin{tikzpicture}[scale=0.75]
\begin{axis}[axis lines=middle, y=15cm, x=15cm, xtick={0,...,3}, axis equal,grid=both, xlabel=$X_i$, ylabel=$X_j$, every axis x label/.style={
    at={(ticklabel* cs:1)},
    anchor=west},
	every tick/.style={
        black,
        thick
      },
every axis y label/.style={
    at={(ticklabel* cs:1.05)},
    anchor=north west},
	every tick/.style={
        black,
        thick
      }
]
\addplot[color=black, mark=] coordinates{(2,3) (2,0)};
\addplot[color=black, mark=] coordinates{(3,3) (3,0)};
\addplot [color=black,mark=,fill=gray,
                    fill opacity=0.5] coordinates {
            (2, 3)
            (3, 3)
            (3, 0)
	    (2,0)
};
\addplot[color=black, mark=] coordinates{(0,1) (4,1)};
\addplot[color=black, mark=] coordinates{(0,2) (4,2)};
\addplot [color=black,mark=,fill=gray,
                    fill opacity=0.5] coordinates {
            (0, 1)
            (0, 2)
            (2, 2)
	    (2,1)
	    (3, 1)
            (3, 2)
            (4, 2)
	    (4,1)

};
\end{axis}
\end{tikzpicture}\caption{An example of a $\TTHREE$-constraint:
	$(2\leq X_i\leq 3)\vee (1\leq X_j\leq 2)$.}\label{fig:type3}
\end{minipage}
\end{figure}

Let us proceed by formally defining \RDTN{s} and \RDTP{s}.
\figref{fig:type2} and \figref{fig:type3} (above) illustrate an example of a $\texttt{t}_2$-constraint and $\texttt{t}_3$-constraint (respectively).
\begin{mydefinition}[\RDTN{s}, \RDTP{s}~\cite{Kumar2006,Kumar05}]
A {\em Restricted Disjunctive Temporal Network}~(\RDTN)~$\N$ is a pair $(\T,\C)$,
	where $\T$ is a set of time-points and $\C = \C_{\TONE}\cup\C_{\texttt{t}_2}\cup\C_{\texttt{t}_3}$ is a set of
{\em restricted disjunctive temporal constraints} over $\T$, each being either one of the following three types:
\begin{enumerate}
\item[($\TONE$)]: $(Y-X\leq w_{X,Y})$, where $X,Y\in \T$ and $w_{X,Y}\in\reali$;
\item[($\texttt{t}_2$)]: $\bigvee_{i=1}^{k}(l_i\leq X\leq u_i)$, where $X\in \T$ and $l_i,u_i\in \reali$ for every $i=1, \ldots, k$;
\item[($\texttt{t}_3$)]: $(l_1\leq X\leq u_1) \vee (l_2\leq Y\leq u_2)$, where $X,Y\in \T$ and $l_i,u_i\in \reali$ for $i=1,2$.
\end{enumerate}
An \RDTN is \textit{consistent} if it admits a \emph{feasible schedule},
\ie some $s: \T\mapsto \reali$ satisfying all of the disjunctive temporal constraints in $\C$.
The {\em Restricted Disjunctive Temporal Problem} (\RDTP) is that of determining whether a given \RDTN is consistent or not.
\end{mydefinition}

Notice that $\TONE$-constraints do coincide with simple temporal constraints of \STN{s}.

We assume w.l.o.g. that the disjuncts of any $\TTWO$-constraint are arranged in ascending order of the end points of their corresponding intervals,
\ie $l_i<l_{i+1}\wedge u_i<u_{i+1}$ $\forall i$, whenever $\bigvee_{i=1}^{k}(l_i\leq X\leq u_i)\in\C_{\texttt{t}_2}$;
these natural orderings on the interval domains of the time-points will be referred to as as their \emph{nominal ordering}.
For any $\tau\in\{1,2,3\}$, $|\C_{\texttt{t}_\tau}|$ denotes the number of $\texttt{t}_\tau$-constraints
  (\ie the cardinality of $\C_{\texttt{t}_\tau}$, not the encoding length).
Also, for any $c_X\in \C_{\texttt{t}_2}$, $|c_X|$ denotes the number of disjuncts of $c_X$,
and $d_{\C_{\texttt{t}_2}}\triangleq \sum_{c_X\in\C_{\texttt{t}_2}} |c_X|$.
Finally, let us fix a total ordering on the time-points, \ie $\T=\{T_1, \ldots, T_k\}$,
this induces an ordering on the pair of disjuncts in any $\TTHREE$-constraint; so, provided,
$c=(l_1\leq X_i\leq u_1)\vee (l_2\leq X_j\leq u_2)\in\C_{\texttt{t}_3}$, for some $i<j$, then,
$d'\triangleq (l_1\leq X_i\leq u_1)$ and $d''\triangleq (l_2\leq X_j\leq u_2)$ are
called the \emph{first} and the \emph{second} disjunct of $c$ (respectively).

As mentioned in the introduction, Kumar showed in~\cite{Kumar05, Kumar2006} that RDTPs are solvable in deterministic
strongly polynomial time by reducing them to CRCs~\cite{DEVILLE1999}.

\subsection{Hyper Temporal Networks}
In order to study the tractability frontier of \RDTP{s}, we shall consider the \HTN model which is grounded on directed hypergraphs as defined next.
\begin{mydefinition}[\cite{CominPR17}]
A \emph{directed hypergraph} $\H$ is a pair $(\T,\A)$, where $\T$ is the set of nodes, and $\A$ is the set of \emph{hyperarcs}.
Each hyperarc $A\in\A$ is either \emph{multi-head} or \emph{multi-tail}:
\begin{itemize}
\item A \emph{multi-head} hyperarc $A=(t_A, H_A, w_A)$ has a distinguished node $t_A$, called the \emph{tail} of $A$, and a non-empty set
$H_A\subseteq V\setminus\{t_A\}$ containing the \emph{heads} of $A$;
to each head $v\in H_A$, it is associated a \emph{weight} $w_A(v)\in\RR$, which is a real number (unless otherwise specified).
\figref{fig:head_hyperarc} depicts a possible representation of a multi-head hyperarc:
the tail is connected to each head by a dashed arc labeled by the name of the hyperarc and the weight associated to the considered head.

\item A \textit{multi-tail} hyperarc $A=(T_A, h_A, w_A)$ has a distinguished node $h_A$, called the \emph{head} of $A$, and a non-empty set
$T_A\subseteq V\setminus\{h_A\}$ containing the \emph{tails} of $A$;
to each tail $v\in T_A$, it is associated a \emph{weight} $w_A(v)\in\RR$, which is a real number (unless otherwise specified).
\figref{fig:tail_hyperarc} depicts a possible representation of a multi-tail hyperarc:
the head is connected to each tail by a dotted arc labeled by the name of the hyperarc and weights.
\end{itemize}
\end{mydefinition}
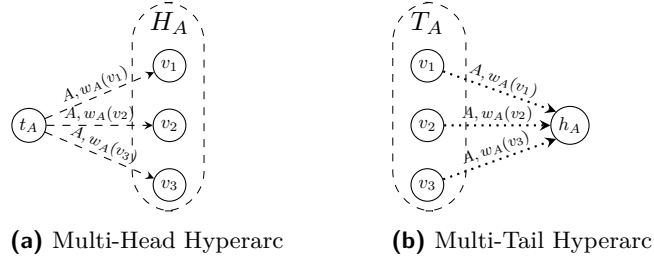
\begin{figure}[!h]
\centering
\begin{minipage}[]{.35\linewidth}
\begin{tikzpicture}[arrows=->,scale=.7,node distance=.5 and 2]
    \node[node,xshift=1ex,label={above, yshift=.5ex:$H_A$}] (v1) {$v_1$};
    \node[node,below=of v1] (v2) {$v_2$};
    \node[node,below=of v2] (v3) {$v_3$};
	\node[node,left=of v2] (u) {$t_A$};
% 	\node[smallLabel,anchor=north east] at (v3.south -| u) {a) Hyperarc $a$.};
   	\draw[>=stealth, multiHead] (u) to node[timeLabel,above,sloped] {$A, w_A(v_1)$} (v1);%
   	\draw[>=stealth, multiHead] (u) to node[timeLabel,above,sloped] {$A, w_A(v_2)$} (v2);%
   	\draw[>=stealth, multiHead] (u) to node[timeLabel,above,sloped] {$A, w_A(v_3)$} (v3);%
 %%%%%%%%%%%%%
  \draw[dashed, ultra thin, rounded corners=15pt] (-.55,1.2) rectangle (.8,-2.75);
\end{tikzpicture}\subcaption{Multi-Head Hyperarc}\label{fig:head_hyperarc}
\end{minipage}
\begin{minipage}[]{.35\linewidth}
\begin{tikzpicture}[arrows=->, scale=.7,node distance=.5 and 2]
    \node[node,label={above, yshift=.5ex:$T_A$}] (v1) {$v_1$};
    \node[node,below=of v1] (v2) {$v_2$};
    \node[node,below=of v2] (v3) {$v_3$};
	\node[node,right=of v2] (u) {$h_A$};
% 	\node[smallLabel,anchor=north east] at (v3.south -| u) {a) Hyperarc $a$.};
   	\draw[>=stealth,multiTail] (v1) to node[timeLabel,above,sloped] {$A, w_A(v_1)$} (u.north west);%
   	\draw[>=stealth, multiTail] (v2) to node[timeLabel,above,sloped] {$A, w_A(v_2)$} (u);%
   	\draw[>=stealth,multiTail] (v3) to node[timeLabel,above,sloped] {$A, w_A(v_3)$} (u.south west);%
    %%%%%%%%%
    \draw[dashed, ultra thin, rounded corners=15pt] (-.65,1.2) rectangle (.725,-2.75);
\end{tikzpicture}\subcaption{Multi-Tail Hyperarc}\label{fig:tail_hyperarc}
\end{minipage}
\caption{Hyperarcs in Hyper Temporal Networks.}\label{fig:hyperarcs}
\end{figure}

The \emph{cardinality} of a hyperarc $A\in \A$ is $|A|\triangleq  |H_A\cup \{t_A\}|$ if $A$ is multi-head,
and $|A| \triangleq |T_A\cup\{h_A\}|$ if $A$ is multi-tail; when $|A|=2$, then $A=(u, v, w)$ is a standard arc.
The \textit{order} and \textit{size} of a directed hypergraph $(\T,\A)$
are $|\T|$ and $m_\A \triangleq \sum_{A\in \A} |A|$~(respectively).

\begin{mydefinition}[\GHTN~\cite{CominPR17}]
A \emph{general-\HTN} is a directed hypergraph $\H = (\T,\A)$ where each node $X\in \T$ represents a time-point,
and each multi-head\slash multi-tail hyperarc stands for a set of temporal
distance constraints between the tail\slash head and the heads\slash tails.
\end{mydefinition}

In general-\HTN{s}, an hyperarc is \textit{satisfied} when at least one of its distance constraints is satisfied.
Then, a \HTN{} is \textit{consistent}
when it is possible to assign a value to each time-point so that all of its hyperarcs are satisfied.
More formally, in the \HTN model the consistency-checking problem is the following decision problem.
\newcommand{\savefootnote}[2]{\footnote{\label{#1}#2}}
\newcommand{\repeatfootnote}[1]{\textsuperscript{\ref{#1}}}
\begin{mydefinition}[\GHTP~\cite{CominPR17}]
Given a general-\HTN \mbox{$\H=(\T,\A)$}, the \emph{General Hyper Temporal Problem} (\GHTP) is
that of deciding whether or not there exists a schedule \mbox{$s:\T \rightarrow \RR$}
such that, for every hyperarc $A\in\A$, the following hold:
\begin{itemize}
\item if $A=(t,h,w)$ is a standard arc, then: $s(h)-s(t)\leq w$;
\item if $A=(t_A, H_A, w_A)$ is a multi-head hyperarc, then:
	$s(t_A) \geq \min_{v\in H_A} \{s(v) - w_A(v) \}$;
\item if $A=(T_A, h_A, w_A)$ is a multi-tail hyperarc, then:
$s(h_A) \leq \max_{v\in T_A} \{s(v) + w_A(v) \}$.
\end{itemize}
\end{mydefinition}

Any such schedule $s$ is called \textit{feasible}.
  A \HTN that admits at least one feasible schedule is called \textit{consistent}.

Comparing the consistency of \HTN{s} with the consistency of \STN{s},
the most important aspect of novelty is that, while in a distance graph of \STN{s} each arc
represents a distance constraint and all such constraints have to be satisfied by any feasible schedule,
in a \HTN each hyperarc represents a disjunction of one or more distance constraints
and a feasible schedule has to satisfy at least one of such distance constraints for each hyperarc.

Let us survey some interesting properties about the consistency-checking problem above.
The first one is that any integer-weighted \HTN admits an integer-valued
  feasible schedule when it is consistent, as stated in the following proposition.
\begin{proposition}[\cite{CominPR17}]\label{prop:int_sched}
Let $\H=(\T,\A)$ be an integer-weighted\footnote{Integer-weighted $\HTN$
  means that $w_A(v)\in\mathbb{Z}$ for every $A\in\A$ and $v\in\T$ for which $w_A(v)$ is defined.} and consistent general-$\HTN$.
Then $\H$ admits an integer feasible schedule $s:\T \rightarrow \{-T,-T+1, \ldots, T-1, T\}$, where $T = \sum_{A\in\A, v\in \T} |w_A(v)|$.
\end{proposition}

%\delCar{The second interesting property is that deciding \GTNC is \NP-complete, as proved in the following theorem.}
\noindent The following theorem states that \GHTP is \NP-complete, in a strong sense.
\begin{mytheorem}[\cite{CominPR17}]\label{Teo:npcompleteness}
\GHTP is an \NP-complete problem even if the input instances $\H=(V, \A)$
are restricted to satisfy $w_A(\cdot) \in\{-1, 0, 1\}$ and $|H_A|, |T_A|\leq 2$ for every $A\in\A$.
\end{mytheorem}

As observed in \cite{CominPR17}, Theorem~\ref{Teo:npcompleteness} motivates the study of consistency problems on \HTN{s} having either
only multi-head or only multi-tail hyperarcs. In the former case, the consistency-checking problem is called \HHTP,
while in the latter it is \THTP; as stated in Theorem~\ref{Teo:MainAlgorithms},
the complexity of checking these two problems turns out to be lower than that for \DTP{s},
\ie both $\HHTP, \THTP \in \NP\cap\coNP$, instead of being $\NP$-complete.

So it's worth considering the following specialized notion of consistency for \HTN{s}.
\begin{mydefinition}[\HHTP]
Given a multi-head \HTN \mbox{$\H=(\T,\A)$},
	the \HHTP problem is that of deciding whether or not there exists a schedule \mbox{$s:\T \rightarrow \RR$} such that:
\[s(t_A) \geq \min_{v\in H_A} \{s(v) - w_A(v) \},\quad \forall A\in\A.\]
%Any such a schedule \mbox{$s:V \rightarrow \R$} is called \textit{feasible}.
%A \HTN that admits at least one feasible schedule is called \textit{consistent}.
\end{mydefinition}

The tightest currently known worst-case time complexity upper-bound for solving (integer-weighted) \HHTP{s} was established
  in~\cite{CominPR17} and it is expressed in the following theorem.

\begin{mytheorem}[\cite{CominPR17}]\label{Teo:MainAlgorithms}
The following proposition holds on (integer-weighted, multi-head) \HTN{s}.
There exists an $O\big((|\T|+|\A|)\cdot m_{\A}\cdot W\big)$ pseudo-polynomial time algorithm for checking \HHTP;
given any \HTN $\H=(\T, \A)$, if $\H$ is consistent the same algorithm also returns an integer-valued feasible schedule $s:\T\rightarrow \interi$ of $\H$;
otherwise, it returns a negative certificate in the form of a negative \emph{hypercycle} (cfr. Appendix~A for more details on that).

Above, $W\triangleq \max_{A\in\A, v\in H_A} |w_A(v)|$ is the maximum absolute value among the weights.
\end{mytheorem}

Concluding this section we recall that the two problems \HHTP and \THTP are actually inter-reducible, \ie one can check any one of the
two models in $f(m,n,W)$-time whenever there's an $f(m,n,W)$-time procedure for checking the consistency of the other one.
\begin{mytheorem}[\cite{CominPR17}]
\HHTP and \THTP are inter-reducible by means of $\log$-space, linear-time, local-replacement reductions.
\end{mytheorem}
Thus, Theorem~\ref{Teo:MainAlgorithms} extends to multi-tail \HTN{s} (\ie they're checkable in pseudo-poly time).

\section{Faster Deterministic Algorithm for $\TTWO$DTPs}\label{sect.Type2-TP_Algo}
This section offers a deterministic quadratic time algorithm for solving temporal problems having only $\{\TONE, \TTWO\}$-constraints, as defined below.

The same algorithm will be leveraged to solve RDTPs fastly, later on in Section~\ref{sect.RDTP_ALGO}.
\begin{mydefinition}
Any \RDTN $\N=(\T, \C_{\TONE}\cup \C_{\TTWO} \cup \C_{\TTHREE})$ having $\C_{\TTHREE}=\emptyset$ is called \emph{$\TTWO$DTN}.

So, $\TTWO$DTNs are denoted simply as $(\T, \C_{\TONE}\cup \C_{\TTWO})$.
The corresponding temporal problem, \ie \emph{$\TTWO$DTP}, is that of determining whether a given $\TTWO$DTN is consistent or not.
\end{mydefinition}

One possible solution to $\TTWO$DTPs is Kumar's reduction from RDTPs to CRCs~\cite{Kumar05}.
Our solution, named \TTP, employs kind of a value-iteration approach in which all  are initially set to zero
and then updated monotonically upwards by necessary arc relaxations --
  this is somehow reminiscent of the BF-VI algorithm for STPs mentioned in Section~\ref{sect.Background}.
Indeed, given a $\TTWO$DTN $\N_{\TTWO} = (\T,\C_{\TONE}\cup \C_{\TTWO})$,
   we firstly solve the \STP $\N_{\TONE} = (\T,\C_{\TONE})$ (\eg with BF-VI).
If $\N_{\TONE}$ is consistent, the returned least feasible schedule
$\hat{\varphi}_\N$ provides an initial candidate, the next step in mind being that of satisfying all the $\TTWO$-constraints.
For this, recall that the disjuncts of any $c_{\TTWO}\in\C_{\TTWO}$ are arranged according to their nominal ordering,
so that we can try to satisfy any given $c_{\TTWO}$ by iteratively picking the next (\ie in ascending order) unsatisfied disjunct of $c_{\TTWO}$
and by enforcing its lower-bound constraint in an auxiliary \STN as if it were a $\TONE$-constraint.
While there's an unsatisfied $\TTWO$-constraint $c_{\TTWO}$,
the current candidate schedule is thus increased by the least necessary amount satisfying both $c_{\TTWO}$ and the whole $\C_{\TONE}$.
It turns out that this can be done efficiently by performing $|\C_{\TTWO}|$ calls to the Dijkstra shortest paths algorithm~\cite{Dijkstra1959}.
In order to show this, let us point out two key facts (\ie Lemma~\ref{lem:slack}~and~\ref{lem:update}).
\begin{mylemma}\label{lem:slack}
Let $\N = (\T,\C_{\TONE})$ be any \STN, and let $\varphi, \varphi'$ be any pair of schedules of $\N$.

Let $\N^{\varphi}$ be the \STN \emph{reweighted} according to the reduced-costs weight transformation~$w^{\varphi}$
(\ie each weight $w_{X,Y}$ in $\N$ is simply replaced by $w^{\varphi}_{X,Y}$),
let $\N^{\varphi'}$ be the same \wrt $w^{\varphi'}$.
Let $\delta^{\varphi}_X:\T\rightarrow \mathbb{Z}\cup\{+\infty\}$
be the length of the shortest path in $\N^{\varphi}$ from any $T\in\T$ to $X\in \T$,
and let $\delta^{\varphi'}_X$ be the same \wrt $\N^{\varphi'}$.

Then for every $T,X\in\T$, either $\delta^{\varphi'}_X(T)$ and $\delta^{\varphi}_X(T)$ are both $+\infty$, or the following holds:
\[ \delta^{\varphi'}_X(T) - \delta^{\varphi}_X(T) =  \big(\varphi'(T)-\varphi(T)\big) - \big(\varphi'(X)-\varphi(X)\big).\]
\end{mylemma}
\begin{proof}
  Let $T,X\in\T$ (arbitrarily), w.l.o.g. $X$ is reachable from $T$ in $\N$ (otherwise, $\delta^{\varphi'}_X(T)$ and $\delta^{\varphi}_X(T)$ are both $+\infty$).
Consider any path $p_{T,X}$ from $T$ to $X$ in $\N^{\varphi}$, \ie for some $k\geq 0$:
  \[ p_{T,X}\triangleq (T=T_0, T_1, T_2, \ldots, T_k = X),
    \text{ having total weight } w^\varphi_{p_{T,X}}\triangleq \sum_{i=0}^{k-1} w^\varphi_{T_i, T_{i+1}} \text{ in } \N^{\varphi}.\]
  Then, the following holds by telescoping:
  \begin{align*}
    w^\varphi_{p_{T,X}} &= \big(w_{T_0,T_1} - \varphi(T_1)+\varphi(T_0)\big) + \big(w_{T_1,T_2}-\varphi(T_2)+\varphi(T_1)\big) + \ldots \\
     & \hspace{35ex} \ldots + \big(w_{T_{k-1},T_k}-\varphi(T_k)+\varphi(T_{k-1})\big) \\
           &= \varphi(T_0) - \varphi(T_k) + \sum_{i=0}^{k-1} w_{T_i,T_{i+1}} = \varphi(T) - \varphi(X) + w_{p_{T,X}}. \tag{1}
  \end{align*}
Thus, provided $\delta_{X}(T)$ is the shortest path distance from $T$ to $X$ in the original network $\N$, we have:
  \begin{align*}\hspace{-4.5ex}
    \delta^\varphi_{X}(T) &= \min\left\{w^\varphi_{p_{T,X}} \mid p_{T,X}\text{ is a path from $T$ to $X$ in $\N$} \right\} & \text{(by def. of $\delta^\varphi_{X}$)} \\
    &= \min\left\{\varphi(T) - \varphi(X) + w_{p_{T,X}} \mid p_{T,X}\text{ is any path from $T$ to $X$ in $\N$} \right\} & \text{(by (1))} \\
    &= \varphi(T) - \varphi(X) + \min\left\{w_{p_{T,X}} \mid p_{T,X}\text{ is any path from $T$ to $X$ in $\N$} \right\} & \text{($\varphi$ is constant here)}\\
    &= \varphi(T) - \varphi(X) + \delta_{X}(T). & \text{(by def. of $\delta_{X}$)}
  \end{align*}
For the same reason, $\delta^{\varphi'}_{X}(T) = \varphi'(T) - \varphi'(X) + \delta_{X}(T)$. Therefore,
\begin{align*}
  \delta^{\varphi'}_X(T) - \delta^{\varphi}_X(T) &= \big(\varphi'(T) - \varphi'(X)+ \delta_{X}(T)\big) - \big(\varphi(T) - \varphi(X)+ \delta_{X}(T)\big) \\
  &= \big(\varphi'(T)-\varphi(T)\big) - \big(\varphi'(X)-\varphi(X)\big).
\end{align*}
This concludes the proof.
\end{proof}

\begin{mylemma}\label{lem:update}
Let $\N = (\T,\C_{\TONE})$ be any \STN, and let $\hat{\varphi}$ be the least feasible schedule of $\N$.
Fix some $X\in\T$ and some real value $l_X\geq \hat{\varphi}(X)$.
Let $\N' = (\T', \C'_{\TONE})$ be the auxiliary \STN obtained by introducing a corresponding lower-bound $\TONE$-constraint over $X$, \ie
\[
  \T'\triangleq \T\cup\{z\}, \\
  \C'_{\TONE}\triangleq \C_{\TONE}\cup\big\{(z-T\leq 0)\mid T\in\T\big\}\cup \big\{(z-X\leq -l_X)\big\}.
\]
Let $\N^{\hat{\varphi}}$ be the \STN reweighted according to the reduced-costs weight transformation $w^{\hat{\varphi}}$,
  and let $\delta^{\hat{\varphi}}_X(T)$ be the length of the shortest path in $\N^{\hat{\varphi}}$ from (any) $T\in\T$ to $X$.

Then, for every $T\in\T$, the least feasible schedule $\hat{\varphi}'$ of $\N'$ is given by:
\[ \hat{\varphi}'(T) = \hat{\varphi}(T) + \max\big(0, l_X - \hat{\varphi}(X) - \delta^{\hat{\varphi}}_X(T) \big).\]
\end{mylemma}
\begin{proof}
%Since $l_X\geq \hat{\varphi}(X)$ and $\delta^{\hat{\varphi}}_X(X)=0$, then
%$\hat{\varphi}'(X)=$

Let w.l.o.g. $\hat{\varphi}'(z)=0$. In order to become feasible for $\N'$ we claim, for every $T\in \T$,
that the least feasible schedule $\hat{\varphi}(T)$ must be increased by at least
$\max\big(0, l_X - \hat{\varphi}(X) - \delta^{\hat{\varphi}}_X(T) \big)$ time units (because of the lower-bound constraint $(z-X\leq -l_X)\in \C'_{\TONE}$).
Indeed, for any $T\in \T$ that reaches $X$ in $\N$, the $\TONE$-constraint $(X-T\leq \delta_X(T))$ (which is induced by telescoping all of
the $\TONE$-constraints along any shortest path from $T$ to $X$) must be satisfied.
On the other hand, by Lemma~\ref{lem:slack} (applied to $\hat\varphi$ and to the
  anywhere-zero\footnote{The anywhere-zero schedule $\zeta$ is that defined as, $\zeta(T)=0$ for every $T\in\T$.} schedule),
  it holds $\hat{\varphi}(X)-\hat{\varphi}(T)=\delta_X(T)-\delta^{\hat{\varphi}}_X(T)$.
This can be seen as follows: if $\hat{\varphi}(T)$ is kept fixed, then $\hat{\varphi}(X)$ can be increased by at most $\delta^{\hat{\varphi}}_X(T)$ time units
	without breaking the induced constraint $(X-T\leq \delta_X(T))$.
Here, $\hat{\varphi}(X)$ must be increased by at least $l_X-\hat{\varphi}(X)$ time units in order to satisfy $(z-X\leq -l_X)\in \C'_{\TONE}$, so
  $\hat{\varphi}(T)$ must be increased by at least the amount said above.

Next, we claim this increase also preserves feasability, \ie it is the \emph{least feasible increase}.
For ease of notation, let $f(z)\triangleq 0$ and $f(T) \triangleq \hat{\varphi}(T) + \max\big(0, l_X - \hat{\varphi}(X) - \delta^{\hat{\varphi}}_X(T) \big)$ $\forall\, T\in\T$.

In order to prove that $f$ satifies all the constraints in $\C_{\TONE}$, pick any $(B-A\leq w_{A,B})\in \C_{\TONE}$.
By hypothesis, it holds:
\begin{equation}
  \hat\varphi(B)-\hat\varphi(A)\leq w_{A,B}.
\end{equation}
For the sake of the argument, let us define: $\Delta_{A,B} \triangleq \big(f(B) - \hat\varphi(B)\big) - \big(f(A) - \hat\varphi(A)\big)$.

So, the following holds:
  \begin{equation}f(B)-f(A) = \hat\varphi(B) - \hat\varphi(A) + \Delta_{A,B}.\end{equation}
Then, either one of the following two cases holds:
\begin{itemize}
 \item If $l_X - \hat\varphi(X) \leq \delta^{\hat{\varphi}}_{X}(B)$, then $f(B) = \hat\varphi(B)$, so $\Delta_{A,B}\leq 0$.
  Therefore,
\begin{align*}
 f(B)-f(A) & \leq \hat\varphi(B) - \hat\varphi(A) & \text{(by (2))} \\
     & \leq w_{A,B}. & \text{(by (1))}
\end{align*}
 \item If $l_X - \hat\varphi(X) > \delta^{\hat{\varphi}}_{X}(B)$,
 it is easy to check that $\Delta_{A,B}\leq \delta^{\hat{\varphi}}_{X}(A) - \delta^{\hat{\varphi}}_{X}(B)$.

By definition of $\delta^{\hat{\varphi}}_{X}$ and since $(B-A\leq w_{A,B})\in \C_{\TONE}$,
 then $\delta^{\hat{\varphi}}_{X}(A)\leq \delta^{\hat{\varphi}}_{X}(B) + w^{\hat\varphi}_{A,B}$.
Therefore,
\begin{align*}
f(B)-f(A) & \leq \hat\varphi(B) - \hat\varphi(A) + \delta^{\hat{\varphi}}_{X}(A) - \delta^{\hat{\varphi}}_{X}(B) \\
		& \leq \hat\varphi(B) - \hat\varphi(A) + w^{\hat\varphi}_{A,B}  = w_{A,B}.
\end{align*}
\end{itemize}
So, in either case, $f(B)-f(A)\leq w_{A,B}$.

Finally, clearly $f(X)=l_X$, so $(z-X\leq -l_X)\in \C'_{\TONE}$ is also satisfied.

This proves $f$ is a feasible schedule of $\N'$. All in, it is the least feasible, \ie $f=\hat\varphi'$.
\end{proof}

With this two facts in mind, the description of \TTP can now proceed more smoothly.

Recall that, firstly, the \STN $\N_{\TONE}$ is checked.
If $\N_{\TONE}$ is already inconsistent, so it is $\N_{\TTWO}$;
  otherwise, $\hat{\varphi}_\N$ is the least feasible schedule of $\N_{\TONE}$.
So, $w^{\hat{\varphi}_\N}\geq 0$ for every constraint in~$\C_{\TONE}$.
Now, for each target node $X\in \T$, the Dijkstra algorithm on input $(\N^{\hat{\varphi}_\N}, X)$ computes $\delta^{\hat{\varphi}_\N}_X(T)$.
The whole distance matrix $\{\delta^{\hat{\varphi}_\N}_X(T)\}_{T\in\T, X\in \T}$ is computed here, and kept stored in memory.
Multiple-sources single-target shortest paths are needed, actually, but these can be easily computed with the traditional Dijkstra's algorithm,
\eg just reverse the direction of all arcs in the input network and treat the single-target node as if it were a single-source.
What follows aims, if there's still an unsatisfied $\TTWO$-constraint $c_{\TTWO}\in\C_{\TTWO}$, at increasing the candidate
schedule $f$ by the least necessary amount satisfying both $c_{\TTWO}$ and the whole $\C_{\TONE}$. So, let us initialize~$f\leftarrow \hat{\varphi}_\N$.
Then the following iterates.

While $\exists$ some $X\in \T$ and $c_X=\bigvee_{i=1}^{k}(l_i\leq X\leq u_i)\in \C_{\TTWO}$ s.t. $f(X)$ doesn't satisfy $c_X$:

if $f(X)>u_k(=\max_i u_i)$, then $\N_{\TTWO}$ is inconsistent (see~Theorem~\ref{thm:TTP_correctness});
  otherwise, let $i^*$ be the smallest $i\in [1,k]$ such that $f(X)<l_i$.
By Lemma~\ref{lem:slack} and given $f$, then $\delta^{f}_X$ is given by:
  \[ \delta^{f}_X(T)\leftarrow
  \delta^{\hat\varphi_0}_X(T) +  \big(f(T)-\hat\varphi_0(T)\big) - \big(f(X)-\hat\varphi_0(X)\big),\; \forall\; T\in\T. \tag{\texttt{rule}-$\delta$}\]
So, by Lemma~\ref{lem:update}, the following updating rule:
  \[ f(T) \leftarrow f(T) + \max\big(0, l_{i^*} - f(X) - \delta^{f}_X(T)\big),\; \forall\; T\in\T.\tag{\texttt{rule}-$f$}\]
  yelds the least feasible schedule for the next auxiliary \STN $\N'_{\TONE}$ obtained by adding the new lower-bound $\TONE$-constraint $(z-X\leq -l_{i^*})$.
At each iteration of the while-loop $\N'_{\TONE}$ is enriched with an additional lower-bound $\TONE$-constraint as above.
So, $\N'_{\TONE}$ has $|\T'|=|\T|+1$ time-points ($z$ included) and at most $|\C'_{\TONE}|\leq |\C_{\TONE}|+|\C_{\TTWO}|$ $\TONE$-constraints
(one $\TONE$-constraint per $c_{\TTWO}\in\C_{\TTWO}$ is enough, as for each $c_{\TTWO}$ only its greatest lower-bound counts).
If the while-loop completes without ever finding $\N_{\TTWO}$ to be inconsistent (because, eventually,
$f(X)>u_k(=\max_i u_i)$ for some $X\in\T$ at some point),
then the last updating of $f$ yelds the least feasible schedule of $\N_{\TTWO}$ (as shown below in~Theorem~\ref{thm:TTP_correctness}).
This concludes the description of \TTP.

Notice that, during the whole computation, the scheduling values can only increase monotonically upwards -- like in a value-iteration process.
\begin{mytheorem}\label{thm:TTP_correctness}
\TTP is correct, \ie on any input $\TTWO$DTN $\N_{\TTWO} = (\T,\C_{\TONE}\cup \C_{\TTWO})$,
it returns a feasible schedule $\varphi:\T\rightarrow\reali$, if $\N_{\TTWO}$ is consistent;
otherwise, it recognizes $\N_{\TTWO}$ as inconsistent.
\end{mytheorem}
\begin{proof}
  Let $\iota=0, 1, 2, \ldots, \iota_h$ be all the iterations of the while-loop of \TTP,
  where $\iota_h$ is assumed to be the last iteration where the updating $\texttt{rule-}f$ is applied.

  For every iteration $\iota\in [1, \iota_h]$, the auxiliary \STN $\N'^{(\iota)}_{\TONE}$ is formally defined as:
  \[ \N'^{(\iota)}_{\TONE}\triangleq (\T\cup\{z\}, \C'^{(\iota)}_{\TONE}), \text{ where } z \text{ is the \emph{zero time-point}, and ... }\]
  \[ \C'^{(\iota)}_{\TONE} \triangleq \C_{\TONE}\cup \big\{(z-T\leq 0)\mid T\in\T\big\}
    \cup \big\{ (z-X^{(\gamma)}\leq -l^{(\gamma)}_{i^*})\mid 1 \leq \gamma \leq \iota\big\},\]
  where, for all $\gamma\leq\iota$, $X^{(\gamma)}$ is the (unique) $X\in \T$ appearing
  in some $\TTWO$-constraint that is considered at the while-loop's $\gamma$-th iteration,
  and $l^{(\gamma)}_{i^*}$ is its corresponding lower-bound.

  Also, let $f^{(\iota)}$ be the candidate schedule as updated by $\texttt{rule-}f$ during the $\iota$-th iteration.

By applying Lemma~\ref{lem:slack}~and~\ref{lem:update} repeatedly,
  for each iteration $\iota$, it holds that $f^{(\iota)}$ is the least feasible schedule of $\N'^{(\iota)}_{\TONE}$.
This is the key invariant at the heart of \TTP.

Concerning actual correctness, firstly, assume that \TTP recognizes $\N_{\TTWO}$ as inconsistent.

If $\N_{\TONE}$ was already inconsistent (cfr.~Theorem~\ref{thm:main_stn}), so $\N_{\TTWO}$ is too.
Otherwise, the inconsistency of $\N_2$ really holds because of these two facts jointly:
(i) the key invariant mentioned above; and,
(ii) at the end of the while-loop, it must be
	$f(X)>u_k(=\max_i u_i)$ for some $\TTWO$-constraint $c_X=\bigvee_{i=1}^{k}(l_i\leq X\leq u_i)\in \C_{\TTWO}$.
Indeed notice that, by (i), no possible feasible schedule $g<f$ can be neglected (discarded)
  during the upward monotone (value-iteration like) updates of the schedules;
and, by (ii), no possible schedule $g\geq f$ can ever satisfy $c_X\in C_2$. So, $\N_{\TTWO}$ is really inconsistent.

Secondly, assume that $\N_{\TTWO}$ is recognized as consistent, by returning a schedule $f^{(\iota_h)}$.

Since \TTP can do that only after the above while-loop completes, the exit condition of the latter
  ensures that $f^{(\iota_h)}$ satisfies every constraint in $\C_{\TTWO}$.
Moreover, the key invariant implies that $f^{(\iota_h)}$ is the least feasible schedule of $\N'^{(\iota_h)}_{\TONE}$,
  so that $f^{(\iota_h)}$ satisfies all of the $\TONE$-constraints in $\C_{\TONE}$.
These two combined, $f^{(\iota_h)}$ is the least feasible schedule~of~$\N_{\TTWO}$.
So, $\N_{\TTWO}$ is indeed consistent.
\end{proof}

The next result asserts that \TTP always halts in time polynomial in the input size.
\begin{mytheorem}\label{thm:ttp_complexity}
Suppose that \TTP runs on input $\TTWO$DTN $\N_{\TTWO} = (\T,\C_{\TONE}\cup \C_{\TTWO})$.

Then, \TTP halts in time
  $O\big(|\T|\cdot |\C_{\TONE}| + |\C_{\TTWO}|\cdot (|\C_{\TONE}| + |\T|\cdot \log |\T|) + |\T|\cdot d_{\C_{\TTWO}}\big)$.
\end{mytheorem}
\begin{proof}
Solving the \STP $\N_{\TONE} = (\T,\C_{\TONE})$ with BF-VI takes $O(|\T|\cdot |\C_{\TONE}|)$ time (cfr. Theorem~\ref{thm:main_stn}).
Computing the shortest paths distance matrix $\{\delta^{\hat{\varphi}_\N}_X(T)\}_{T\in\T, X\in \T}$
  takes $|\C_{\TTWO}|$ calls to the Dijkstra algorithm (one per $X\in\T$ participating in some $\TTWO$-constraint),
  so, $O(|\C_{\TTWO}|\cdot(|\C_{\TONE}| + |\T|\cdot \log|\T|))$ total time.
Checking the while-loop exit condition (\ie wether there exists some unsatisfied $c_X\in \C_{\TTWO}$),
can be done in $O(|\T|\cdot d_{\C_{\TTWO}})$ total time (because there are at most $d_{\C_{\TTWO}}$ iterations and each check can be done in $O(|\T|)$ time).
At each iteration of the while-loop, applying $\texttt{rule-}\delta$ and $\texttt{rule-}f$ to all $T\in\T$ takes $O(|\T|)$ time per iteration,
and we have at most $d_{\C_{\TTWO}}$ of them; so, notice that it takes only $O(1)$ time per single application of the rules.

Therefore, the overall time complexity of \TTP on any input $\N_{\TTWO}=(\T,\C_{\TONE}\cup\C_{\TTWO})$~is:
\[ \texttt{Time}_{\TTP}(\N_{\TTWO}) = O\big(|\T|\cdot |\C_{\TONE}| + |\C_{\TTWO}|\cdot (|\C_{\TONE}| + |\T|\cdot \log |\T|) + |\T|\cdot d_{\C_{\TTWO}}\big).\]
This is a strongly polynomial time, \ie not depending on the magnitude of the arc weights.
\end{proof}
\section{Faster Deterministic Algorithm for \RDTP{s}}\label{sect.RDTP_ALGO}
With our brand new $\TTWO$DTPs algorithm in mind, let us now focus on solving \RDTP{s} fastly.
Given an input \RDTP $\N = (\T,\C_{\TONE}\cup \C_{\TTWO} \cup \C_{\TTHREE})$,
we firstly solve the $\TTWO$DTP $\N_{\TTWO} = (\T,\C_{\TONE}\cup \C_{\TTWO})$ with \TTP (cfr. Section~\ref{sect.Type2-TP_Algo}).
If $\N_{\TTWO}$ is already inconsistent, we're done as $\N$ is too.
Otherwise, the key idea is that of checking the consistency of all the $\TTHREE$-constraints by
making one single reduction call to the 2-SAT problem (which can be solved in linear-time~\cite{Aspvall1979}).

For this reason, the universe of boolean variables is $\{x_c\}_{c\in\C_{\TTHREE}}$, \ie we have one variable per $c\in\C_{\TTHREE}$.
Let $d',d''$ be the first and second disjunct of any given $c\in\C_{\TTHREE}$ (respectively),
the intended interpretation being that $x_c$ is $\texttt{true}$ iff $d'$ is satisfied (and $d''$ can be anything),
whereas $x_c$ is $\texttt{false}$ iff $d'$ is unsatisfied and $d''$ is satisfied.

The 2-CNF formula $\textsc{Cl}_{\N}$ is built as follows.
Basically, for each $c\in\C_{\TTHREE}$ and each disjunct $d$ of $c$, we enforce the binding requirement of
satisfying all the temporal constraints in $\{d\}\cup\C_{\TONE}\cup \C_{\TTWO}$, and we check whether this implies that some other
disjunct $\tilde{d}$ of any other $\TTHREE$-constraint $\tilde{c}\neq c$ becomes unsatisfiable as a consequence.
More precisely, we check whether satisfying $\{d\}\cup\C_{\TONE}\cup \C_{\TTWO}$ implies that some weight $\tilde{u}$ must become a \emph{strict} lower-bound for
the scheduling value of some $\tilde{X}\in\T$ that appears in some other $\TTHREE$-disjunct $\tilde{d}=(\tilde{l}\leq \tilde{X}\leq \tilde{u})$.
This is formalized in Definition~\ref{def:clause} (below).
If that is the case, a binary clause asserting the above
implication\footnote{Here, recall the rule of material implication $p\rightarrow q\leftrightarrow \neg p \vee q$.} is added to $\textsc{Cl}_{\N}$.
Let us formally describe the details of this construction.
\begin{mydefinition}\label{def:clause}
Given any \RDTP $\N = (\T,\C_{\TONE}\cup \C_{\TTWO} \cup \C_{\TTHREE})$, initially $\textsc{Cl}_{\N}$ is an empty set of binary clauses.
For each $\TTHREE$-constraint of $\N$, \eg for each $c=d'_c \vee d''_c\in \C_{\TTHREE}$ where
$d'_c=(l_1\leq X_i\leq u_1)$ and $d''_c=(l_2\leq X_j\leq u_2)$, some $i<j$, $\textsc{Cl}_{\N}$ is populated as~follows:
\begin{enumerate}
\item Consider the $\TTWO$DTP $\N[d'_c]_{\TTWO}$ in which $d'_c$ is added to $\C_{\TONE}$ as a pair of $\TONE$-constraints, \ie
\begin{align*}
  \N[d'_c]_{\TTWO}\triangleq \Big( \T\cup\{z\}, \big(\C_{\TONE} & \cup \{ (z-X_i\leq -l_1), (X_i-z\leq u_1) \} \\
    &\cup \{z-T\leq 0\mid T\in\T\}\big) \; \cup \; \C_{\TTWO} \Big).
\end{align*}
If $\N[d'_c]_{\TTWO}$ is consistent, let $ \hat\varphi[d'_c]$ be its least feasible schedule; otherwise, add the unary clause $\neg x_c$ to $\textsc{Cl}_{\N}$.
For each $\tilde{c}\neq c$ in $\C_{\TTHREE}$, \eg $\tilde{c}=(\tilde{l}_1\leq X_{\tilde{i}}\leq \tilde{u}_1) \vee (\tilde{l}_2\leq X_{\tilde{j}}\leq \tilde{u}_2)\in\C_{\TTHREE}$,
\begin{itemize}
\item if $ \hat\varphi[d'_c](X_{\tilde{i}})>\tilde{u}_1$ then add the implication
  $x_c\Rightarrow \neg x_{\tilde{c}}$ (\ie clause $\neg x_c \vee \neg x_{\tilde{c}}$) to $\textsc{Cl}_{\N}$;
\item if $ \hat\varphi[d'_c](X_{\tilde{j}})>\tilde{u}_2$ then add the implication
  $x_c\Rightarrow x_{\tilde{c}}$ (\ie clause $\neg x_c \vee x_{\tilde{c}}$) to $\textsc{Cl}_{\N}$.
\end{itemize}
\item Consider the $\TTWO$DTP $\N[d''_c]_{\TTWO}$ in which $d''_c$ is added to $\C_{\TONE}$ (similarly as above).
If $\N[d''_c]_{\TTWO}$ is consistent, let $ \hat\varphi[d''_c]$ be its least feasible schedule; otherwise, add the unary clause $x_c$ to $\textsc{Cl}_{\N}$.
Again, for each $\TTHREE$-constraint $\tilde{c}\neq c$ of $\N$, \eg $\tilde{c}=(\tilde{l}_1\leq X_{\tilde{i}}\leq \tilde{u}_1) \vee (\tilde{l}_2\leq X_{\tilde{j}}\leq \tilde{u}_2)$:
 if $ \hat\varphi[d''_c](X_{\tilde{i}})>\tilde{u}_1$ then add the implication $\neg x_c\Rightarrow \neg x_{\tilde{c}}$ (\ie clause $x_c \vee \neg x_{\tilde{c}}$) to $\textsc{Cl}_{\N}$;
and if $ \hat\varphi[d''_c](X_{\tilde{j}})>\tilde{u}_2$ then add the implication $\neg x_c\Rightarrow x_{\tilde{c}}$ (\ie the clause $x_c \vee x_{\tilde{c}}$) instead.
\end{enumerate}
\end{mydefinition}
So, if the 2-SAT problem instance $\textsc{Cl}_{\N}$ is unsatisfiable, the input \RDTP $\N$ is inconsistent.
Otherwise, for every $c=d'\vee d''\in\C_{\TTHREE}$ we get at least one feasible $\TTWO$DTP:
either $\N[d'_c]_{\TTWO}$, which is
 related to the first disjunct $\{d'\}\cup\C_{\TONE}\cup \C_{\TTWO}$; or $\N[d''_c]_{\TTWO}$,
  which is related to the second $\{d''\}\cup\C_{\TONE}\cup \C_{\TTWO}$
(according to whether $x_c$ is $\texttt{true}$ or not in the satisfying assignment of $\textsc{Cl}_{\N}$).
Then we compute the pointwise-maximum schedule taken among all of~those.~Formally,
\begin{mydefinition}\label{def:max_sched}
Let $\phi:\{x_c\}_{c\in\C_{\TTHREE}}\rightarrow\{\texttt{true}, \texttt{false}\}$ be any satisfying assignment of $\textsc{Cl}_{\N}$.

For every $c=d'_c \vee d''_c\in \C_{\TTHREE}$, let us define:
\[ d^\phi_c \triangleq \left\{
\begin{array}{ll}
d'_c, & \text{ if } \phi(x_c)=\texttt{true}; \\
d''_c, &  \text{ otherwise.}
\end{array}
\right.\;\;\;\;\;\;\;\;
\text{ then, } \;\;\;\;\;\;\;\; \check{\varphi}_\N(T) \triangleq \max_{c\in\C_{\TTHREE}} \hat\varphi[d^\phi_c](T),\;\; \forall T\in\T,\]
where $\hat\varphi[d^\phi_c]$ denotes the least feasible schedule of the consistent $\TTWO$DTP $\N[d^\phi_c]_{\TTWO}$.
\end{mydefinition}
The above pointwise-maximum schedule $\check{\varphi}_\N$ turns out to be feasible for the input RDTP~$\N$, as we show next.
It is assumed we are given an \RDTP $\N$ for which $\textsc{Cl}_{\N}$ is satisfiable.
\begin{proposition}\label{prop:varphi_c1}
Given $\N$ as above, the schedule $\check{\varphi}_\N$ satisfies every $c\in \C_{\TONE}$.
\end{proposition}
\begin{proof}
Let $c_{\TONE}=(Y-X\leq w_{X,Y})\in\C_{\TONE}$ be any $\TONE$-constraint, some $X,Y\in\T$ and $w\in\reali$.
Pick any $c_Y^* \in \arg\max_{c\in\C_{\TTHREE}} \hat\varphi[d^\phi_c](Y)$.
Clearly, $\max_{c\in\C_{\TTHREE}} \hat\varphi[d^\phi_c](X) \geq \hat\varphi[d^\phi_{c_Y^*}](X)$. Therefore:
\begin{align*}
\check{\varphi}_\N(Y) - \check{\varphi}_\N(X) &= \max_{c\in\C_{\TTHREE}} \hat\varphi[d^\phi_c](Y) - \max_{c\in\C_{\TTHREE}} \hat\varphi[d^\phi_c](X) \\
	 	       &\leq \hat\varphi[d^\phi_{c_Y^*}](Y) - \hat\varphi[d^\phi_{c_Y^*}](X) \leq w,
\end{align*}
where the very last inequality holds because $\hat\varphi[d^\phi_{c_Y^*}]$ is feasible for $(\T, \C_{\TONE})$.
So, $\check{\varphi}_\N$ satisfies~$c_{\TONE}$.
\end{proof}

\begin{proposition}\label{prop:varphi_c2}
Given $\N$ as above, the schedule $\check{\varphi}_\N$ satisfies every $c\in \C_{\TTWO}$.
\end{proposition}
\begin{proof}
Let $c_{\TTWO}=\bigvee_{i=1}^{k}(l_i\leq X\leq u_i)\in \C_{\TTWO}$ be any $\TTWO$-constraint, some $X\in \T$, $l_i,u_i\in \reali$.
Pick any $c_X^* \in \arg\max_{c\in\C_{\TTHREE}} \hat\varphi[d^\phi_c](X)$.
By definition $\hat\varphi[d^\phi_{c_X^*}]$ is a feasible schedule of $\N[d^\phi_{c_X^*}]_{\TTWO}$, thus it is feasible for $(\T, \C_{\TTWO})$ too. Therefore,
\[ l_q\leq \hat\varphi[d^\phi_{c_X^*}](X)\leq u_q, \text{ for some } q\in \{1, \ldots, k\}.\]
Since, $\check{\varphi}_\N(X)=\varphi[d^\phi_{c_X^*}](X)$, then $\check{\varphi}_\N(X)\in [l_q,u_q]$ for the same $q$. So, $\check{\varphi}_\N$ satisfies $c_{\TTWO}$.
\end{proof}

\begin{proposition}\label{prop:varphi_c3}
Given $\N$ as above, the schedule $\check{\varphi}_\N$ satisfies every $c\in \C_{\TTHREE}$.
\end{proposition}
\begin{proof}
Let $c_{\TTHREE}=(l_1\leq X\leq u_1) \vee (l_2\leq Y\leq u_2)\in \C_{\TTHREE}$ be any $\TTHREE$-constraint,
  some $X,Y\in \T$, $X<Y$ and $l_1,l_2,u_1,u_2\in \reali$.
Assume w.l.o.g. $\phi(x_{c_{\TTHREE}})=\texttt{true}$. Then, $l_1\leq \hat\varphi[d^\phi_{c_{\TTHREE}}](X)\leq u_1$.

If ${c_{\TTHREE}} \in \arg\max_{c\in\C_{\TTHREE}} \hat\varphi[d^\phi_c](X)$, then
  $\check{\varphi}_\N(X)= \hat\varphi[d^\phi_{c_{\TTHREE}}](X)\in [l_1,u_1]$; so, $\check{\varphi}_\N$ would satisfy~$c_{\TTHREE}$.
Otherwise, ${c_{\TTHREE}} \not\in \arg\max_{c\in\C_{\TTHREE}} \hat\varphi[d^\phi_c](X)$, and assume $\check{\varphi}_\N(X)\not\in [l_1,u_1]$ towards a contradiction.
Pick any $c_X^* \in \arg\max_{c\in\C_{\TTHREE}} \hat\varphi[d^\phi_c](X)$. All these hypotheses combined:
\[\check{\varphi}_\N(X)= \hat\varphi[d^\phi_{c_X^*}](X) > u_1.\]
Therefore, $\phi$ must satisfy either $p\Rightarrow \neg x_{c_{\TTHREE}}$ or $\neg p\Rightarrow \neg x_{c_{\TTHREE}}$,
for some boolean variable $p$ (where the actual case depends on the actual value of $d^\phi_{c_X^*}$).
Since $\phi$ satisfies either $p$ or $\neg p$, then $\phi$ must satisfy $\neg x_{c_{\TTHREE}}$; \ie $\phi(x_{c_{\TTHREE}})=\texttt{false}$.
This is absurd, as we assumed $\phi(x_{c_{\TTHREE}})=\texttt{true}$.

The proof of the other case, in which $\phi(x_{c_{\TTHREE}})=\texttt{false}$ is initially assumed, is symmetric.
So, $\check{\varphi}_\N(X)$ satisfies $c_{\TTHREE}$.
\end{proof}
Let us mention that our algorithm is called \RDTPc, basically, it aims at computing $\hat\varphi_\N$ as above;
  if it fails in that (either because $\N_{\TTWO}$ is already inconsistent
    or $\textsc{Cl}_{\N}$ is unsatisfiable), it recognizes the input \RDTP $\N$ as inconsistent.
Now, we can prove this is correct and fast.
\begin{mytheorem}\label{thm:varphi_feasible}
  \RDTPc is correct, \ie on any RDTN $\N = (\T,\C_{\TONE}\cup \C_{\TTWO} \cup \C_{\TTHREE})$,
    it returns a feasible schedule $\hat\varphi_\N:\T\rightarrow\reali$, if $\N$ is consistent;
    otherwise, $\N$ is recognized as inconsistent.
\end{mytheorem}
\begin{proof}
Recall that $\N$ is recognized as inconsistent only if $(\T,\C_{\TONE}\cup\C_{\TTWO})$ is already inconsistent
  or if the 2-SAT problem instance $\textsc{Cl}_{\N}$ is unsatisfiable.
In the former case, since $(\T,\C_{\TONE}\cup\C_{\TTWO})$ is inconsistent, so it is $\N$.
In the latter, by construction of $\textsc{Cl}_{\N}$,
it is not possible to satisfy all the constraints in
  $\C_{\TONE}\cup\C_{\TTWO}\cup \C_{\TTHREE}$ (otherwise, the reader can check, it would've been possible to
  construct a satisfying assignment for $\textsc{Cl}_{\N}$, straightforwardly);
so, $\N$ is really inconsistent.

On the other side, by Propositions~\ref{prop:varphi_c1},~\ref{prop:varphi_c2}~and~\ref{prop:varphi_c3}, schedule $\hat\varphi_\N$ is really feasible for $\N$.
  \end{proof}

%  This completes the proof of Theorem~\ref{thm:varphi_feasible}.\end{proof}

The next result asserts that the halting time is strongly polynomial in the input size.
\begin{mytheorem}
Let \RDTPc run on any input RDTP $\N = (\T,\C_{\TONE}\cup \C_{\TTWO}\cup C_3)$.

Its always halts within time
  $O\big(|\T|\cdot |\C_{\TONE}| + |\C_{\TTWO}|\cdot (|\C_{\TONE}| + |\T|\cdot \log |\T|) + |\T|\cdot d_{\C_{\TTWO}}\cdot|\C_{\TTHREE}| + |\C_{\TTHREE}|^2\big)$.
\end{mytheorem}
\begin{proof}
By Theorem~\ref{thm:ttp_complexity}, $(\T,\C_{\TONE}\cup \C_{\TTWO})$ takes
  $O\big(|\T|\cdot |\C_{\TONE}| + |\C_{\TTWO}|\cdot (|\C_{\TONE}| + |\T|\cdot \log |\T|) + |\T|\cdot d_{\C_{\TTWO}}\big)$ time to be checked.
Using that solution as an initial candidate, solving the two $\TTWO$DTPs $\N[d'_c]_{\TTWO}$ and $\N[d''_c]_{\TTWO}$, for each $c\in \C_{\TTHREE}$ where $c=d'_c \vee d''_c$,
it takes $O\big(|\T|\cdot |\C_{\TONE}| + |\C_{\TTWO}|\cdot (|\C_{\TONE}| + |\T|\cdot \log |\T|) + |\T|\cdot d_{\C_{\TTWO}}\cdot|\C_{\TTHREE}|\big)$ total time.
Next, for each $c,\tilde{c}\in \C_{\TTHREE}$ such that $\tilde{c}\neq c$,
eventually adding the corresponding clauses to $\textsc{Cl}_{\N}$ takes $O(1)$ time per clause;
so, $\textsc{Cl}_{\N}$ is built in total time $O(|\C_{\TTHREE}|^2)$.
Since $|\textsc{Cl}_{\N}|=O(|\C_{\TTHREE}|^2)$, solving the 2-SAT problem on input $\textsc{Cl}_{\N}$ takes time $O(|\C_{\TTHREE}|^2)$
  (\eg with the algorithm of~\cite{Aspvall1979}).
Finally, computing $d^\phi_c$ and $\check{\varphi}_\N$ takes $O(|\T|\cdot |\C_{\TTHREE}|)$ time.
All in, the above mentioned time complexity of \RDTPc follows.
\end{proof}

\section{\NP-completeness of Multi-Tail \& Multi-Head $\TTHREE$HyTPs}
This section enquiries the tractability frontier of RDTPs by considering \HTN{s}~\cite{CominPR17},
where the basic idea is that of blending the two models together and see what happens to the complexity of the corresponding temporal problems.
Two restricted kinds of disjunctive temporal problems, \TTTTTNC and \TTTHTNC, are both proven to be \NP-complete.
The former problem is that of deciding whether a multi-tail {$\TTHREE$}\HTN (\ie a temporal network in which the
 constraints can be modeled only by multi-tail hyperarcs and by $\TTHREE$-constraints) is consistent or not.
The latter, \TTTHTNC, is the same as the former but considers multi-head hyperarcs instead. Let us now focus on \TTTTTNC.
\begin{mytheorem}\label{Teo:npcompleteness_tail}
\TTTTTNC is \NP-complete in a strong sense, \ie even if the input $(\T, \A\cup\C_{\texttt{t}_3})$ are restricted
  to satisfy $w_A(\cdot) \in\{-1, 0, 1\}$, $|T_A|\leq 2$ for every $A\in\A$,
    and every $\TTHREE$-constraint $(l_i\leq X\leq u_i) \vee (l_j\leq Y\leq u_j)\in \C_{\texttt{t}_3}$ has all zero-valued lower/upper-bounds.
\end{mytheorem}
\begin{proof}
We claim that if $\H=(\T, \A\cup\C_{\texttt{t}_3})$ is an integer-weighted and consistent multi-tail {$\TTHREE$}\HTN,
 it admits an integer-valued feasible schedule $s:\T\rightarrow\{-T, \ldots, T\}$
	where $T = \sum_{A\in\A, v\in V} |w_A(v)| + \sum_{c\in\C_{\texttt{t}_3}, c=(l_1\leq X\leq u_1)\vee (l_2\leq Y\leq u_2)} (|l_1|+|u_1|+|l_2|+|u_2|)$.
Indeed let $s$ be a feasible schedule (integer-valued or not) of $\H$,
 and consider the projection \HTN $\H^s\triangleq (\T, \A')$,
for	$ \A'\triangleq \A\cup\bigcup_{c\in\C_{\texttt{t}_3}} A^s_c$,
	where for every $c=(l_1\leq X\leq u_1)\vee (l_2\leq Y\leq u_2)\in\C_{\texttt{t}_3}$ we pick the following pair of $\TONE$-constraints:

		$ A^s_c\triangleq \left\{
					\begin{array}{ll}
						\big\{(Z-X\leq -l_1), (X-Z\leq u_1)\big\}, & \text{ if } l_1\leq s(X)\leq u_1;\\
						\big\{(Z-Y\leq -l_2), (Y-Z\leq u_2)\big\}, & \text{ otherwise.}
				\end{array}\right.$

By construction of $\H^s$, $s$ is a feasible for \HTN $\H^s$.
So, by Proposition~\ref{prop:int_sched}, $\H^s$ admits an integer-valued feasible schedule $s'$ bounded by $-T$ and $+T$ as above.
By contruction of $\H^s$, $s'$ is feasible for $\H$ too.

Moreover, any such integer-valued feasible schedule can be verified in strongly polynomial time \wrt the size of the input; hence, \TTTTTNC is in \NP.

To show that the problem is \NP-hard, we describe a reduction from 3-SAT.

Let us consider a boolean 3-CNF formula with $n\geq 1$ variables and $m\geq 1$ clauses:

$\varphi(x_1, \ldots, x_n) = \bigwedge_{i=1}^m (\alpha_i \vee \beta_i \vee \gamma_i)$,
where $C_i = (\alpha_i \vee \beta_i \vee \gamma_i)$ is the $i$-th clause of $\varphi$
and  each $\alpha_i,\beta_i,\gamma_i\in \{x_j, \overline{x}_j\mid 1\leq j\leq n\}$ is either a positive or a negative literal.

We associate to $\varphi$ a multi-tail {$\TTHREE$}\HTN $\H_{\varphi}=(\T, \A\cup \C_{\texttt{t}_3})$,
where each boolean variable $x_i$ occurring in $\varphi$ gets represented by two time-points, $x_i$ and $\overline{x}_i$.
$\T$ also contains a time-point $z$ that represents the reference
initial time-point for $\H_{\varphi}$, \ie the first time-point that has to be executed at time zero.
Moreover, for each pair $x_i$ and $\overline{x}_i$, $\H_{\varphi}$ contains:
%\begin{itemize}
%\item %\figref{FIG:Var_i}: one with multi-head $\{x_i,\overline{x}_i\}$ and tail in $z$

a multi-tail hyperarc with tails $\{x_i,\overline{x}_i\}$, both weighted $-1$, and head in $z$.
%\item

a $\TTHREE$-constraint $\big((0\leq x_i\leq 0) \vee (0\leq \overline{x}_i\leq 0)\big)\in\C_{\texttt{t}_3}$.
%\end{itemize}
If $\H_{\varphi}$ is consistent, the multi-tail hyperarc and the $\TTHREE$-constraint associated to $x,\neg x$ assures that
$\H_{\varphi}$ admits an integer feasible schedule $s$ (as we mentioned above) such that $s(x_i)$ and $s(\overline{x}_i)$
are coherently set with values in $\{0,1\}$. In this way, $s$ is forced to encode a truth assignment on the $x_i$'s.

The \HTN $\H_{\varphi}$ contains also a time-point $C_j$ for each clause $C_j$ of $\varphi$;
each $C_j$ is connected by a multi-tail hyperarc with head in $C_j$ and tails over
the literals occurring in $C_j$ and by two standard and opposite arcs with time-point $z$ as displayed in \figref{fig:gadgets} (right).
This assures that if $\H_{\varphi}$ admits a feasible schedule $s$,
then $s$ assigns scheduling time $1$ at least to one of the time-point representing the literals connected with the multi-tail hyperarc.

\figref{fig:gadgets} depicts the gadgets.
\begin{figure}[tb]
\begin{tikzpicture}[arrows=->,scale=1,node distance=2 and 2]
	\node[node, label={below:$[0]$}] (zero) {$z$};
	\node[node, above right=of zero] (nX) {$\overline{x}_i$};
	\node[node, above left=of zero] (X) {$x_i$};
 	%arcs
	%%%%%%% X's arcs %%%%%%
	\draw[] (zero) to [bend left=40] node[below] {$1$} (X);
	\draw[] (X) to [bend left=45] node[above] {$0$} (zero);
	\draw[>=,dashed, thick,sloped] (zero) to [bend right=10] node[above,xshift=.55ex,yshift=-.55ex] {\footnotesize $(0,0)$, \tiny $\TTHREE$} (X);
	\draw[dotted, thick] (X) to [bend right=15] node[below] {$-1$} (zero);
	%%%%%%% nX's arcs %%%%%
	\draw[] (zero) to [bend left=45] node[above] {$1$} (nX);
	\draw[] (nX) to [bend left=40] node[below] {$0$} (zero);
	\draw[>=, dashed, thick,sloped] (zero) to [bend left=10] node[above,yshift=-.55ex] {\footnotesize $(0,0)$, \tiny $\TTHREE$} (nX);
	\draw[dotted, thick] (nX) to [bend left=15] node[below] {$-1$} (zero);
\end{tikzpicture}\label{FIG:Var_i}
\begin{tikzpicture}[arrows=->,scale=.65,node distance=1.5 and 2]
	\node[node,label={above:$[1]$}] (one) {$C_j$};
	\node[node,below =of one] (beta) {$\beta_j$};
	\node[node,left=of beta] (alpha) {$\alpha_j$};
	\node[node,right=of beta] (gamma) {$\gamma_j$};
 	\node[node,label={below:$[0]$}, below=of beta] (zero) {$z$};
 	\coordinate (fakeL) at ($(alpha.west)+(-.3,0)$);
 	\coordinate (fakeR) at ($(gamma.east)+(.3,0)$);
	%arcs
	%%%%%%% zero/one arcs %%%%%
% 	\path [bend left=70,through point=(fakeL.south west),through point=(fakeL.north west)] (zero) edge (one);
 	\draw[>=] (zero) to [bend left=45] node[left] {$+1$} (fakeL.north);
 	\draw[] (fakeL) to [bend left=45] (one);
  	\draw[>=] (one) to [bend left=45] node[above] {$-1$} (fakeR.south);
  	\draw[] (fakeR) to [bend left=45] (zero);
% 	\draw [bend left=35,through point=(fakeR)] (one) edge (zero.south east);
	%%%%%%% alpha's arcs %%%%%%
	\draw[] (alpha) to [bend right=20] node[timeLabel,below] {} (zero);
	\draw[] (zero) to [bend right=20] (alpha);
	\draw[dashed, dotted] (alpha) to [bend right=5] (zero);
	\draw[dashed] (zero) to [bend right=5] (alpha);
	\draw[dotted, thick] (alpha) to [bend right=25] node[timeLabel,below] {$0$} (one.south west);
	%%%%%%% beta's arcs %%%%%
	\draw[] (zero) to [bend left=25] (beta);
	\draw[] (beta) to [bend left=25] (zero);
	\draw[dotted] (beta) to [bend right=10] (zero);
	\draw[dashed] (zero) to [bend right=10] (beta);
	\draw[dotted, thick] (beta) to [] node[timeLabel,right] {$0$} (one.south);
	%%%%%%% gamma's arcs %%%%%%
	\draw[] (gamma) to [bend left=20] (zero);
	\draw[] (zero) to [bend left=20] (gamma);
	\draw[dotted] (gamma) to [bend left=5] (zero);
	\draw[dashed] (zero) to [bend left=5] (gamma);
	\draw[dotted, thick] (gamma) to [bend left=25] node[timeLabel,below] {$0$} (one.south east);
\end{tikzpicture}
\caption{Variable and clause gadgets (at left and right, respectively) used in Theorem~\ref{Teo:npcompleteness_tail}.}\label{fig:gadgets}
\end{figure}
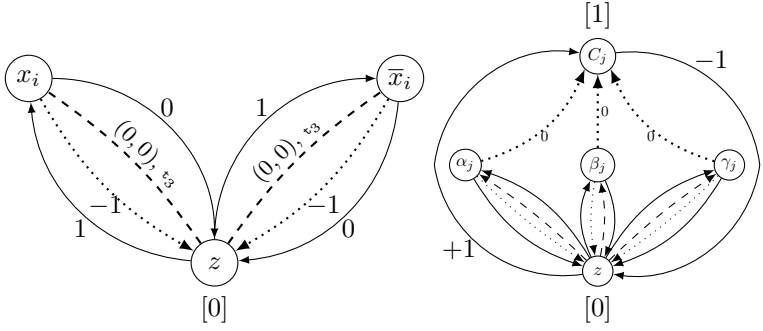
A more formal definition of $\H_{\varphi}$ is given in Appendix~A.

The reader can check that $|\T|=1+2n+m=O(m+n)$, $m_{\A}=O(m+n)$, $|\C_{\texttt{t}_3}|=O(n)$; therefore, the transformation is linearly bounded.

We next show that $\varphi$ is satisfiable if and only if $\H_{\varphi}$ is consistent.

Any truth assignment $\nu:\{x_1, \ldots, x_n\}\rightarrow \{\texttt{true}, \texttt{false}\}$
satisfying $\varphi$ can be translated into a feasible schedule $s:\T\rightarrow \interi$ of $\H_{\varphi}$ as follows.
For time-point $z$, let $s(z)=0$, and let $s(C_j)=1$ for each $j=1, \ldots, m$; then,
for each $i=1, \ldots, n$, let $s(x_i) = 1$ and $s(\overline{x}_i) = 0$ if the truth value of $x_i$, $\nu(x_i)$, is \texttt{true},
otherwise let $s(x_i) = 0$ and $s(\overline{x}_i) = 1$. It is simple to verify that,
using this schedule $s$, all the constraints comprising each single gadget are satisfied and, therefore, the network is consistent.
So, $\H_{\varphi}$ is consistent.

Vice versa, assume that $\H_{\varphi}$ is consistent. Then, it admits an integer-valued feasible schedule $s$ (as we mentioned above).
After the translation $s(v)\triangleq s(v) - s(z)$, we can assume that $s(z)=0$.
Hence, $s(C_j) = 1$ for each $j=1,\ldots, m$, as enforced by the two standard arcs incident at $C_j$ in the clause gadget,
and $\{s(x_i),s(\overline{x}_i)\} = \{0,1\}$ for each $i=1,\ldots, n$, as enforced by the constraints comprising the variable gadgets.
Therefore, the feasible schedule $s$ can be translated into a truth assignment $\nu:\{x_1, \ldots, x_n\}\rightarrow \{\texttt{true}, \texttt{false}\}$
defined by $\nu(x_i)=\texttt{true}$ if $s(x_i)=1$ (and $s(\overline{x}_i)=0$);
$\nu(x_i)=\texttt{false}$ if $s(x_i)=0$ (and $s(\overline{x}_i)=1$) for every $i = 1, \ldots , n$. So, $\varphi$ is satisfiable.
%Notice that $\nu$ satisfies $\varphi$ since, for every $j = 1, \ldots ,
%m$, it satisfies clause $C_j$ as otherwise $s$ would not satisfy the hyperarc constraint of the $C_j$ clause gadget.

To conclude, we observe that any hyperarc $A\in\A$ of $\H_{\varphi}$ has weights $w_A(\cdot)\in\{-1, 0, 1\}$, size $|A|\leq 3$,
and any $\TTHREE$-constraint $c=(l_i\leq X\leq u_i) \vee (l_j\leq Y\leq u_j)\in \C_{\texttt{t}_3}$ has zero lower and upper-bounds (\ie $l_i=u_i=l_j=u_j=0$).
Since any hyperarc with three tails can be replaced by two hyperarcs each having at most two tails,
the consistency problem remains \NP-Complete even if $|A|\leq 2$ for every $A\in A$.
\end{proof}
In order to prove that \TTTHTNC is also \NP-complete, we could proceed with an argument similar to that of Theorem~\ref{Teo:npcompleteness_tail}.
However, we also observe that the same result follows as an immediate corollary of the following inter-reducibility between the two models.
\begin{mydefinition}
A multi-tail (multi-head) RHyTN is any temporal network in which the
     constraints can be modeled only by multi-tail (multi-head) hyperarcs and by $\{\TTWO, \TTHREE\}$ disjunctive temporal constraints.

The problem of checking whether a given RHyTN is consistent is named RHyTP. Observe,
\end{mydefinition}
\begin{proposition}\label{prop:inter-reducitble-HTNs}
Multi-head and multi-tail RHyTPs are inter-reducible by means of $\log$-space, linear-time, local-replacement reductions.
Particularly, multi-head and multi-tail $\TTHREE$HyTPs are inter-reducible by such reductions. (The proof is in Appendix~A)
\end{proposition}
Therefore, by Proposition~\ref{prop:inter-reducitble-HTNs}, it follows that \TTTHTNC is also strongly \NP-complete.

\section{Pseudo-Polynomial Time Algorithm for $\texttt{t}_2$HyTPs}\label{sect:algo_type2hytp}
We end by studying multi-tail and multi-head $\texttt{t}_2$HyTNs (\ie temporal networks in which
the temporal constraints can be only $\TTWO$ disjunctive temporal constraints and either only multi-tail or multi-head hyperarcs).
It turns out that checking the corresponding temporal problems, \TTTtwoTTNC and \TTTtwoHTNC,
lies in $\textsc{NP}\cap \text{co-}\textsc{NP}$ and admits pseudo-polynomial time algorithms.
By Proposition~\ref{prop:inter-reducitble-HTNs}, it is sufficient to focus on multi-head $\texttt{t}_2$HyTPs only.
The corresponding pseudo-polynomial time algorithm is named \TTHTP, and described below -- notice that it generalizes \TTP.
Given any integer-weighted multi-head $\texttt{t}_2$HyTPs $\H_{\texttt{t}_2} = (\T,\A\cup \C_{\texttt{t}_2})$ in input,
we firstly solve the HyTP $\H = (\T,\A)$ with the VI algorithm of Theorem~\ref{Teo:MainAlgorithms}.
If $\H$ is recognized as inconsistent, the algorithm halts. Otherwise, let $\varphi$ be the least feasible schedule of $\H$.
Then proceed as follows:

While $\exists$ some $X\in \T$ and $c_X=\bigvee_{i=1}^{k}(l_i\leq X\leq u_i)\in \C_{\texttt{t}_2}$ s.t. $\varphi(X)$ doesn't satisfy $c_X$:

If $\varphi(X)>u_k(=\max_i u_i)$, then $\H_{\texttt{t}_2}$ is recognized as inconsistent;
  otherwise, let $i^*$ be the smallest $i\in [1,k]$ such that $\varphi(X)<l_i$.
Firstly, we increase the value of $\varphi(X)$ up to $l_{i^*}$, \ie update $\varphi(X)\leftarrow l_{i^*}$.
Secondly, the VI algorithm of Theorem~\ref{Teo:MainAlgorithms} is invoked on input $(\H, \varphi)$,
so, then, $\varphi$ becomes the schedule returned by that run of VI. The process iterates so on and so forth, and
if the while-loop completes without recognizing $\H_{\texttt{t}_2}$ as inconsistent, $\varphi$ is returned.
The correctness and the time complexity are asserted below. (The proof is in Appendix~A)
\begin{mytheorem}\label{thm:htn_algo}
\TTHTP is correct, \ie running on any integer-weighted multi-head $\texttt{t}_2$HyTP $\H_{\texttt{t}_2} = (\T,\A\cup \C_{\texttt{t}_2})$,
 an integer-valued feasible schedule $\varphi:\T\rightarrow\interi$ is returned, in case $\H_{\texttt{t}_2}$ is consistent;
otherwise, $\H_{\texttt{t}_2}$ is correctly recognized as inconsistent.

Moreover, the corresponding time complexity is pseudo-polynomial, \ie
\begin{align*}\hspace{2.5ex} \texttt{Time}_{\texttt{t}_2\texttt{HyTP()}}(\H_{\texttt{t}_2}) =
   O\big(( & |\T|+|\A|) \cdot m_{\A}\cdot W_{\A,\C_{\texttt{t}_2}} \big), \\
	\text{ where } & W_{\A,\C_{\texttt{t}_2}}\triangleq \max\Big(\max_{A\in \A}\max_{h\in A} |w_A(h)|, \max_{\substack{l_j \text{ appears in any } \\
    \vee_{i=1}^{k}(l_i\leq X\leq u_i)\in \C_{\texttt{t}_2}}} l_j\Big).
\end{align*}
\end{mytheorem}
Finally, since \TTHTP is correct, it is possible to establish the following complexity result.
\begin{mytheorem}\label{thm:np_conp}
  $\TTTtwoTTNC, \TTTtwoHTNC \in \textsc{NP}\cap \text{co-}\textsc{NP}$. (The proof is in Appendix~A)
\end{mytheorem}
%(The proofs of Theorem~\ref{thm:htn_algo}~and~\ref{thm:np_conp} are both in Appendix~A)

%!TEX root = cstnWithDecision.tex
\section{Conclusions and Future Works}\label{sect.CFW}
A deeper combinatorial comprehension on the algorithmics of RDTPs led
 to a new elementary deterministic strongly polynomial time procedure for solving them,
  significantly improving the asymptotic running times suggested by Kumar before.
In future works we'd like to investigate further on possible generalizations/extensions of the proposed algorithms,
   aiming at covering some compatible (or even wider) subclasses of the disjunctive temporal constraints~problem.

%\subparagraph*{Acknowledgements.}
%I want to thank \dots

%% Either use bibtex (recommended),
\bibliography{bibliography}
\newpage
\appendix
\section{Appendix: Omitted Proofs.}
The appendix proceeds by offering additional missing proofs.
\begin{proof}[Proof of Proposition~\ref{prop:inter-reducitble-HTNs}]
We show the reduction from multi-tail to multi-head hypergraphs; the converse direction is symmetric.
Informally, all the arcs are reversed (so that what was multi-tail becomes multi-head),
and, contextually, the time-axis is inverted (to account for the inversion of the direction of all arcs).
Finally, all $\texttt{t}_2$ and $\texttt{t}_3$-constraints are also reversed.

Given a multi-tail R\HTN $\H=(\T,\A\cup\C_{\texttt{t}_2}\cup\C_{\texttt{t}_3})$, we associate to $\H$ a multi-head R\HTN $\H'=(\T, \A'\cup\C'_2\cup\C'_3)$
by reversing all multi-tail hyperarcs, all $\texttt{t}_2$ and $\texttt{t}_3$-constraints.
Formally,
\begin{align*} & {\A}'  \triangleq \Big\{(v, S, w)\mid (S, v, w)\in \A\Big\}, \;\;
  \C'_2\triangleq \Big\{ \bigvee_{i=1}^{k}(-u_i\leq X\leq -l_i) \mid \bigvee_{i=1}^{k}(l_i\leq X\leq u_i)\in\C_{\texttt{t}_2} \Big\}, \\
  & \C'_3\triangleq \Big\{ \big((-u_1\leq X\leq -l_1) \vee (-u_2\leq Y\leq -l_2)\big)
                              \mid \big((l_1\leq X\leq u_1) \vee (l_2\leq Y\leq u_2)\big)\in\C_{\texttt{t}_3} \Big\}.
\end{align*}

We claim that $\H$ is consistent if and only if $\H'$ is consistent.
To prove it, we note that each schedule $s$ for $\H$ can be associated, with a flip of the time direction, to the schedule $s' \triangleq -s$.
Then,  it holds that $s$ is feasible for $\H$ if and only if $s'$ is feasible for $\H'$.
Indeed, $s$ satisfies the constraint represented by an hyperarc $A=(T_A, h_A, w_A)\in\A$, \ie

$
   s(h_A)\leq \max_{v\in T_A} \{s(v) + w_A(v) \}
$,

or, equivalently, $
	-s(h_A)\geq \min_{v\in T_A}\{-s(v) - w_A(v) \}$,
if and only if $s'$ (that is, $-s$) satisfies the constraint represented by the reversed hyperarc $A'=(h_A, T_A, w_A)$, \ie if and only if:

$
   s'(h_A)\geq \min_{v\in T_A}\{s'(v)-w_{A'}(v)\}
$.

Next, $s$ satisfies a $\texttt{t}_2$-constraint $\bigvee_{i=1}^{k}(l_i\leq X\leq u_i)$ \iff
$l_i \leq s(X)\leq u_i$ holds for some $i\in [1,k]$, or equivalently, \iff $ -u_i \leq -s(X)\leq -l_i$; this happens
\iff $s'$ satisfies the constraint represented by the reversed disjunct $(-u_i\leq X\leq -l_i)$,
  \ie \iff $-u_i \leq s'(X)\leq -l_i$.

Finally, $s$ satisfies a $\texttt{t}_3$-constraint $((l_1\leq X\leq u_1) \vee (l_2\leq Y\leq u_2))$ \iff either
  $l_1 \leq s(X)\leq u_1$ or $l_2 \leq s(Y)\leq u_2$, or equivalently, either
    $-u_1 \leq -s(X)\leq -l_1$ or $-u_2 \leq -s(Y)\leq -l_2$; this happens
  \iff $s'$ satisfies the constraint represented either by the reversed disjunct $(-u_1\leq X\leq -l_1)$ or by $(-u_2\leq Y\leq -l_2)$,
    \ie \iff either
      $-u_1 \leq s'(X)\leq -l_1$ or $-u_2 \leq s'(Y)\leq -l_2$.

\end{proof}

\begin{proof}[Formal definition of $\H_{\varphi}$ in the proof of Theorem~\ref{Teo:npcompleteness_tail}]

	More formally, $\H_{\varphi}=(\T, \A\cup \C_{\texttt{t}_3})$ is:
	\begin{itemize}
		\item $\T=\{z\}\cup \{x_i \mid 1\leq i \leq n\}\cup \{\overline{x}_i\mid 1\leq i \leq n\}\cup\{C_j \mid 1\leq j\leq m\}$;
		\item $\A = \bigcup_{i=1}^n \text{Var}_i \cup \bigcup_{j=1}^m\text{Cla}_j$, where:
		\begin{itemize}
			\item $\text{Var}_i=\Big\{ (z, x_i, 1), (x_i, z, 0), (z, \overline{x}_i, 1), (\overline{x}_i, z, 0), \\
				 \big(\{x_i, \overline{x}_i\}, z, [w(x_i), w(\overline{x}_i)] = [-1, -1] \big)\Big\}$.\\
				This is for the variable gadget of $x_i$ as depicted in \figref{fig:gadgets} (left);
			\item $\text{Cla}_j=\Big\{(z, C_j, 1), (C_j, z, -1), \\
						(\{\alpha_j, \beta_j, \gamma_j\}, C_j, [ w(\alpha_j), w(\beta_j), w(\gamma_j) ] = [0,0,0] )\Big \}$.\\
				This defines the clause gadget for clause $C_j = (\alpha_i \vee \beta_i \vee \gamma_i)$ as in \figref{fig:gadgets} (right).
		\end{itemize}
		\item $\C_{\texttt{t}_3} = \bigcup_{i=1}^n \text{Var}'_i$, where:
		\begin{itemize}
			\item $\text{Var}'_i=\big\{ \big((0\leq x_i\leq 0) \vee (0\leq \overline{x}_i\leq 0)\big) \big\}$.\\
				This completes the variable gadget of $x_i$ as depicted in \figref{fig:gadgets} (left);
		\end{itemize}
	\end{itemize}

\end{proof}

\begin{proof}[Proof of $\TTTtwoHTNC, \TTTtwoTTNC \in \textsc{NP}$]
  We claim that if $\H=(\T, \A\cup\C_{\texttt{t}_2})$ is an integer-weighted and consistent multi-tail $\TTWO$\HTN,
  it admits an integer-valued feasible schedule $s:\T\rightarrow\{-T, \ldots, T\}$
  where $T = \sum_{A\in\A, v\in V} |w_A(v)| + \sum_{c\in\C_{\texttt{t}_2}, c=\bigvee_{i=1}^{k}(l_i\leq X\leq u_i)} (|l_i|+|u_i|)$.
  Indeed, let $s$ be a feasible schedule (integer-valued or not) of $\H$,
  and consider the projection~\HTN:
\begin{align*}
  \H^s & \triangleq (\T\cup\{z\}, \A'), \\
  \A' & \triangleq \A\cup\{(z-T\leq 0)\mid T\in\T\}\cup\bigcup_{c\in\C_{\texttt{t}_2}} A^s_c.
\end{align*}
  where for every $c=\bigvee_{i=1}^{k}(l_i\leq X\leq u_i)\in\C_{\texttt{t}_2}$ this pair of $\TONE$-constraints is taken:
  \[ A^s_c\triangleq
  			\big\{(z-X\leq -l_i), (X-z\leq u_i)\mid \text{ for the smallest } i\in\{1, \ldots, k\} \text{ s.t. } l_i\leq s(X)\leq u_i \big\}.\]
  By construction of $\H^s$, $s$ is a feasible for \HTN $\H^s$.
  So, by Proposition~\ref{prop:int_sched}, $\H^s$ admits an integer-valued feasible schedule $s'$ bounded as above.
  By construction of $\H^s$, $s'$ is feasible for $\H$ too.

  Any such integer-valued feasible schedule can be verified in strongly polynomial time \wrt the size of the input,
    simply by checking the actual consistency of each constraint in $\A\cup\C_{\texttt{t}_2}$; hence, \TTTtwoHTNC is in \NP.
  Thus, by Proposition~\ref{prop:inter-reducitble-HTNs}, $\TTTtwoTTNC\in \NP$.
\end{proof}

\begin{mydefinition}[Hypercycle]\label{def:neg_cyc}
We recall from~\cite{CominPR17} that a \emph{hypercycle} $\C_0$ in a \HTN $\H$ is actually
a pair $(S,\C_0)$ with $S\subseteq \T$ and $\C_0 \subseteq \A$ such that:
\begin{enumerate}
\item $S = \cup_{A\in \C_0} A$ and $S\neq\emptyset$;
\item $\forall v\in S$ there exists an unique $A\in \C_0$ such that $t_A = v$.
\end{enumerate}
Every infinite path in a cycle $(S,\C)$ contains, at least, one \textit{finite cyclic sequence} $v_i, v_{i+1}, \ldots, v_{i+p}$,
where $v_{i+p} = v_i$ is the only repeated node in the sequence.
A cycle $(S,\C_0)$ is \textit{negative} if for any finite cyclic sequence $v_1, v_{2}, \ldots, v_{p}$,
it holds that $\sum_{t=1}^{p-1} w_{a(v_t)}(v_{t+1}) < 0$, where $a(v)$ denotes the unique arc $A\in \C_0$ with $t_A = v$ as required in previous item~2.
\end{mydefinition}

\begin{mydefinition}[Certified Least Feasible Schedule (CLFS)]\label{def:cert_leastfeasible}
Given any integer-weighted multi-head \HTN $\H=(\T, \A)$,
a \emph{certified least feasible schedule (CLFS)} for $\H$ is a pair $\varphi_{\text{cert}}\triangleq \{\varphi, \F\}$,
where $\varphi:\T\rightarrow \interi$ is a feasible schedule of $\H$,
and $\F\triangleq \{\C_X\}_{X_\T}$ is a family of hypercycles of $\H$ (which works as a certificate of minimality for $\varphi$, as follows):
for every $X\in\T$, $\C_X$ is a negative hypercycle of the auxiliary \HTN $\H_X$ obtained from $\H$
just by adding one $\TONE$-constraint requiring $X$ to be scheduled strictly before time $\varphi(X)$,
\[ \H_X\triangleq \big(\T\cup\{z\}, \A\cup\{(z-T\leq 0)\mid T\in\T\}\cup\{X-z\leq \varphi(X)-1\}\big).\]
$\varphi_{\text{cert}}$ can be verified in strongly polynomial time,
because negative hypercycles can be checked so (as shown \eg in Lemma~3~of~\cite{CominPR17})
and feasiblity of $\varphi$ can be checked by inspection.
The soundness of CLFSs follows from the proof argument of Proposition~\ref{prop:int_sched}
  (\ie the idea in this proof is -- cfr. Lemma~2 in~\cite{CominPR17} -- to project the feasible HyTN over a conservative graph and
 then, in that setting, to exploit the integrality properties of potentials as prescribed \eg by the Bellman-Ford algorithm) plus
the fact that the universe of feasible schedules of any given multi-head \HTN is closed under pointwise-minimum,
  \ie given two feasible schedules $s_1,s_2$, the schedule $s(u)\triangleq \min(s_1(u), s_2(u))$ $\forall u\in\T,$ is still feasible
(also notice that, in multi-tail  \HTN{s}, the pointwise-maximum works instead).
\end{mydefinition}
With Definitions~\ref{def:neg_cyc}~and~\ref{def:cert_leastfeasible} in mind, we can proceed with the following proof.
\begin{proof}[Proof of $\TTTtwoHTNC, \TTTtwoTTNC \in \text{co-}\textsc{NP}$]
To show $\TTTtwoHTNC\in \text{co-}\NP$, we shall exhibit certificates (of inconsistent networks) which we can verify in strongly polynomial~time.

The basic idea, in order to construct such certificates, is to consider what happens during the execution of algorithm \TTHTP,
assuming the input instance $\H=(\T, \A\cup\C_{\texttt{t}_2})$ is inconsistent.
If the \HTN $\H_0\triangleq (\T, \A)$ is already inconsistent, then it admits a negative hypercycle $\C_0$ (see \eg Theorem~4~in~\cite{CominPR17}).
Moreover, $\C_0$ can be checked in strongly polynomial time (see \eg Lemma~3~in~\cite{CominPR17}),
  so $\C_0$ is already a valid certificate of inconsistency.
Otherwise, let $\varphi_0$ be a CLFS of $\H_0$, then there must exist some $X_0\in \T$ and
$c_{X_0}=\bigvee_{i=1}^{k}(l_i\leq X_0\leq u_i)\in \C_{\texttt{t}_2}$ s.t. $\varphi_0(X_0)$ doesn't satisfy $c_{X_0}$.
If $\varphi_0(X_0)>u_k(=\max_i u_i)$, then $(\varphi_0, c_{X_0})$ is a valid certificate of inconsistency;
otherwise, let $i_0^*$ be the smallest $i\in [1,k]$ such that $\varphi_0(X_0)<l_i$. Let $\H_1$ be the auxiliary \HTN obtained from $\H_0$
just by adding one $\TONE$-constraint requiring $X_0$ to be scheduled at or after time $l_{i_0^*}$:
\[
  \H_1\triangleq \big(\T\cup\{z\}, \A\cup\{(z-T\leq 0)\mid T\in\T\}\cup\{z-X_0\leq - l_{i_0^*}\}\big).
  \]
Then, let $\varphi_1$ be a CLFS of $\H_1$ (notice $\varphi_1$ exists because $\H_0$ was assumed to be feasible).
Again, there must exist some $X_1\in \T$ and $c_{X_1}=\bigvee_{i=1}^{k}(l_i\leq X_1\leq u_i)\in \C_{\texttt{t}_2}$ s.t. $\varphi_1(X_1)$ doesn't satisfy $c_{X_1}$.
Again, if $\varphi_1(X_1)>u_k(=\max_i u_i)$, then $(\varphi_0, c_{X_0}, \varphi_1, c_{X_1})$ is a valid certificate of inconsistency;
otherwise, we can construct yet another auxiliary \HTN $\H_2$ by adding one $\TONE$-constraint requiring $X_1$ to be scheduled at or after time $l_{i_1^*}$,
for appropriate $i_1^*$ defined similarly as before. The construction iterates inductively and, generally,
  it leads to a sequence of the following kind (where the $\{\varphi_i\}_{i=0}^N$ are all CLFSs of the auxiliary \HTN{s}):
\[
  \Big( (\varphi_0, c_{X_0}), (\varphi_1, c_{X_1}), \ldots, (\varphi_N, c_{X_N}) \Big),
\]
where, notice, it's length is at most $N\leq d_{\C_{\texttt{t}_2}}$ (because, for each iteration of the construction,
  one disjunct of some $\TTWO$-constraint is ruled out forever). Each element of the sequence can be verified in strongly polynomial time,
  and the length of the same sequence is strongly polynomial; plus, Theorem~\ref{thm:htn_algo} implies the correctness of these certificates.
  This proves that $\TTTtwoHTNC\in \text{co-}\NP$.
  Thus, by Proposition~\ref{prop:inter-reducitble-HTNs}, $\TTTtwoTTNC\in \text{co-}\NP$ too.
\end{proof}

\begin{proof}[Proof of Theorem~\ref{thm:htn_algo}]
The correctness argument is similar to that for proving correctness of \TTP, the details are simpler
in this case because the only algorithm that is used to update
the schedule $\varphi$ monotonically is the VI algorithm of
Theorem~\ref{Teo:MainAlgorithms} (instead of Bellman-Ford and multiple calls to Dijkstra as it was for \TTP);
indeed, the VI algorithm of Theorem~\ref{Teo:MainAlgorithms} also provides the \emph{least} feasible
schedule in case the input \HTN is consistent, thus a similar (actually simpler) correctness argument still holds.
Also the time complexity of \TTHTP is a direct consequence of the complexity of the VI algorithm of Theorem~\ref{Teo:MainAlgorithms},
where the maximal weight measure $W$ is increased to $W_{\A, \C_{\texttt{t}_2}}$ (as defined above) in order
to take into account the lower-bound constraints (\ie those of type $(z-X\leq -l_{i^*})$) that
are (implicitly) introduced in $\A$ during the main while-loop; notice that,
during the computation, the VI algorithm is invoked on input $(\H, \varphi)$
so that at each iteration the scheduling values are initialized to those of the previous iteration (this ensures that they
are always updated monotonically upwards during the whole computation, thus amortizing the total cost among all iterations). Plus,
at each iteration at least one scheduling value is increased (\ie $\varphi(X)$ is increased to $l_{i^*}$ to satisfy the last $(z-X\leq -l_{i^*})$).
Finally, checking the while-loop's condition takes time $O(d_{\C_{\texttt{t}_2}}\cdot |\T|)$ total time, and since $d_{\C_{\texttt{t}_2}}\leq |\T|\cdot W_{\A, \C_{\texttt{t}_2}}$,
then $O(d_{\C_{\texttt{t}_2}}\cdot |\T|)=O(|\T|^2\cdot W_{\A,\C_{\texttt{t}_2}})$ (which, notice, it is not a bottleneck asymptotically).
So, the $\texttt{Time}_{\texttt{t}_2\texttt{HyTP()}}$ bound holds.
\end{proof}

\end{document}